\ifx\onecol\undefined
\documentclass[journal]{IEEEtran}
\else 
\documentclass[onecolumn,draftcls,12pt]{IEEEtran}
\fi

\usepackage[utf8]{inputenc}
\usepackage{physics}
\usepackage{amsmath}
\usepackage{amssymb}
\usepackage{subfigure}
\usepackage{subfig}
\usepackage{graphicx}
\usepackage{graphics}
\usepackage{color}
\usepackage{psfrag}
\usepackage{cite}
\usepackage{balance}
\usepackage{algorithm}
\usepackage{accents}
\usepackage{amsthm}
\usepackage{bm}
\usepackage{url}
\usepackage{algorithmic}
\usepackage[english]{babel}
\usepackage{multirow}
\usepackage{enumerate}
\usepackage{cases}
\usepackage{stfloats}
\usepackage{dsfont}
\usepackage{soul}
\usepackage{amsfonts}
\usepackage{fancyhdr}
\usepackage{hhline}
\usepackage{array}
\usepackage{booktabs}
\usepackage{framed}
\usepackage{float}
\usepackage{soul}
\usepackage{colortbl}
\usepackage{url}
\usepackage{mathtools}
\usepackage{threeparttable}
\usepackage{booktabs}

\definecolor{lightblue}{rgb}{0.8,0.8,1}

\newcommand{\mb}[1]{{  \mathbf  #1}}

\begin{document}
\newtheorem{theorem}{Theorem}
\newtheorem{acknowledgement}[theorem]{Acknowledgement}
\newtheorem{axiom}[theorem]{Axiom}
\newtheorem{case}[theorem]{Case}
\newtheorem{claim}[theorem]{Claim}
\newtheorem{conclusion}[theorem]{Conclusion}
\newtheorem{condition}[theorem]{Condition}
\newtheorem{conjecture}[theorem]{Conjecture}
\newtheorem{criterion}[theorem]{Criterion}
\newtheorem{definition}{Definition}
\newtheorem{exercise}[theorem]{Exercise}
\newtheorem{lemma}{Lemma}
\newtheorem{corollary}{Corollary}
\newtheorem{notation}[theorem]{Notation}
\newtheorem{problem}[theorem]{Problem}
\newtheorem{proposition}{Proposition}
\newtheorem{solution}[theorem]{Solution}
\newtheorem{summary}[theorem]{Summary}
\newtheorem{assumption}{Assumption}
\newtheorem{example}{\bf Example}
\newtheorem{remark}{\bf Remark}

\newtheorem{thm}{Corollary}[section]
\renewcommand{\thethm}{\arabic{section}.\arabic{thm}}

\def\qed{$\Box$}
\def\QED{\mbox{\phantom{m}}\nolinebreak\hfill$\,\Box$}
\def\proof{\noindent{\emph{Proof:} }}
\def\poof{\noindent{\emph{Sketch of Proof:} }}
\def
\endproof{\hspace*{\fill}~\qed
\par
\endtrivlist\unskip}
\def\endproof{\hspace*{\fill}~\qed\par\endtrivlist\vskip3pt}

\def\E{\mathsf{E}}
\def\eps{\varepsilon}
\def\phi{\varphi}
\def\Lsp{{\boldsymbol L}}
\def\Bsp{{\boldsymbol B}}
\def\lsp{{\boldsymbol\ell}}
\def\Ltsp{{\Lsp^2}}
\def\Lpsp{{\Lsp^p}}
\def\Linsp{{\Lsp^{\infty}}}
\def\LtR{{\Lsp^2(\Rst)}}
\def\ltZ{{\lsp^2(\Zst)}}
\def\ltsp{{\lsp^2}}
\def\ltZt{{\lsp^2(\Zst^{2})}}
\def\ninN{{n{\in}\Nst}}
\def\oh{{\frac{1}{2}}}
\def\grass{{\cal G}}
\def\ord{{\cal O}}
\def\dist{{d_G}}
\def\conj#1{{\overline#1}}
\def\ntoinf{{n \rightarrow \infty}}
\def\toinf{{\rightarrow \infty}}
\def\tozero{{\rightarrow 0}}
\def\trace{{\operatorname{trace}}}
\def\ord{{\cal O}}
\def\UU{{\cal U}}
\def\rank{{\operatorname{rank}}}
\def\acos{{\operatorname{acos}}}

\def\SINR{\mathsf{SINR}}
\def\SNR{\mathsf{SNR}}
\def\SIR{\mathsf{SIR}}
\def\tSIR{\widetilde{\mathsf{SIR}}}
\def\Ei{\mathsf{Ei}}
\def\l{\left}
\def\r{\right}
\def\lb{\left\{}
\def\rb{\right\}}

\setcounter{page}{1}

\newcommand{\eref}[1]{(\ref{#1})}
\newcommand{\fig}[1]{Fig.\ \ref{#1}}

\def\bydef{:=}
\def\ba{{\mathbf{a}}}
\def\bb{{\mathbf{b}}}
\def\bc{{\mathbf{c}}}
\def\bd{{\mathbf{d}}}
\def\bee{{\mathbf{e}}}
\def\bff{{\mathbf{f}}}
\def\bg{{\mathbf{g}}}
\def\bh{{\mathbf{h}}}
\def\bi{{\mathbf{i}}}
\def\bj{{\mathbf{j}}}
\def\bk{{\mathbf{k}}}
\def\bl{{\mathbf{l}}}
\def\bm{{\mathbf{m}}}
\def\bn{{\mathbf{n}}}
\def\bo{{\mathbf{o}}}
\def\bp{{\mathbf{p}}}
\def\bq{{\mathbf{q}}}
\def\br{{\mathbf{r}}}
\def\bs{{\mathbf{s}}}
\def\bt{{\mathbf{t}}}
\def\bu{{\mathbf{u}}}
\def\bv{{\mathbf{v}}}
\def\bw{{\mathbf{w}}}
\def\bx{{\mathbf{x}}}
\def\by{{\mathbf{y}}}
\def\bz{{\mathbf{z}}}
\def\b0{{\mathbf{0}}}

\def\bA{{\mathbf{A}}}
\def\bB{{\mathbf{B}}}
\def\bC{{\mathbf{C}}}
\def\bD{{\mathbf{D}}}
\def\bE{{\mathbf{E}}}
\def\bF{{\mathbf{F}}}
\def\bG{{\mathbf{G}}}
\def\bH{{\mathbf{H}}}
\def\bI{{\mathbf{I}}}
\def\bJ{{\mathbf{J}}}
\def\bK{{\mathbf{K}}}
\def\bL{{\mathbf{L}}}
\def\bM{{\mathbf{M}}}
\def\bN{{\mathbf{N}}}
\def\bO{{\mathbf{O}}}
\def\bP{{\mathbf{P}}}
\def\bQ{{\mathbf{Q}}}
\def\bR{{\mathbf{R}}}
\def\bS{{\mathbf{S}}}
\def\bT{{\mathbf{T}}}
\def\bU{{\mathbf{U}}}
\def\bV{{\mathbf{V}}}
\def\bW{{\mathbf{W}}}
\def\bX{{\mathbf{X}}}
\def\bY{{\mathbf{Y}}}
\def\bZ{{\mathbf{Z}}}

\def\bxi{{\boldsymbol{\xi}}}

\def\sT{{\mathsf{T}}}
\def\sH{{\mathsf{H}}}
\def\cmp{{\text{cmp}}}
\def\cmm{{\text{cmm}}}
\def\WPT{{\text{WPT}}}
\def\lo{{\text{lo}}}
\def\gl{{\text{gl}}}

\def\tT{{\widetilde{T}}}
\def\tF{{\widetilde{F}}}
\def\tP{{\widetilde{P}}}
\def\tG{{\widetilde{G}}}
\def\tbh{{\widetilde{\mathbf{h}}}}
\def\tbg{{\widetilde{\mathbf{g}}}}

\def\mA{{\mathbb{A}}}
\def\mB{{\mathbb{B}}}
\def\mC{{\mathbb{C}}}
\def\mD{{\mathbb{D}}}
\def\mE{{\mathbb{E}}}
\def\mF{{\mathbb{F}}}
\def\mG{{\mathbb{G}}}
\def\mH{{\mathbb{H}}}
\def\mI{{\mathbb{I}}}
\def\mJ{{\mathbb{J}}}
\def\mK{{\mathbb{K}}}
\def\mL{{\mathbb{L}}}
\def\mM{{\mathbb{M}}}
\def\mN{{\mathbb{N}}}
\def\mO{{\mathbb{O}}}
\def\mP{{\mathbb{P}}}
\def\mQ{{\mathbb{Q}}}
\def\mR{{\mathbb{R}}}
\def\mS{{\mathbb{S}}}
\def\mT{{\mathbb{T}}}
\def\mU{{\mathbb{U}}}
\def\mV{{\mathbb{V}}}
\def\mW{{\mathbb{W}}}
\def\mX{{\mathbb{X}}}
\def\mY{{\mathbb{Y}}}
\def\mZ{{\mathbb{Z}}}

\def\cA{\mathcal{A}}
\def\cB{\mathcal{B}}
\def\cC{\mathcal{C}}
\def\cD{\mathcal{D}}
\def\cE{\mathcal{E}}
\def\cF{\mathcal{F}}
\def\cG{\mathcal{G}}
\def\cH{\mathcal{H}}
\def\cI{\mathcal{I}}
\def\cJ{\mathcal{J}}
\def\cK{\mathcal{K}}
\def\cL{\mathcal{L}}
\def\cM{\mathcal{M}}
\def\cN{\mathcal{N}}
\def\cO{\mathcal{O}}
\def\cP{\mathcal{P}}
\def\cQ{\mathcal{Q}}
\def\cR{\mathcal{R}}
\def\cS{\mathcal{S}}
\def\cT{\mathcal{T}}
\def\cU{\mathcal{U}}
\def\cV{\mathcal{V}}
\def\cW{\mathcal{W}}
\def\cX{\mathcal{X}}
\def\cY{\mathcal{Y}}
\def\cZ{\mathcal{Z}}
\def\cd{\mathcal{d}}
\def\Mt{M_{t}}
\def\Mr{M_{r}}
\def\O{\Omega_{M_{t}}}
\newcommand{\figref}[1]{{Fig.}~\ref{#1}}
\newcommand{\tabref}[1]{{Table}~\ref{#1}}

\newcommand{\fb}{\tx{fb}}
\newcommand{\nf}{\tx{nf}}
\newcommand{\BC}{\tx{(bc)}}
\newcommand{\MAC}{\tx{(mac)}}
\newcommand{\Pout}{p_{\mathsf{out}}}
\newcommand{\nnn}{\nn\\}
\newcommand{\FB}{\tx{FB}}
\newcommand{\TX}{\tx{TX}}
\newcommand{\RX}{\tx{RX}}
\renewcommand{\mod}{\tx{mod}}
\newcommand{\m}[1]{\mathbf{#1}}
\newcommand{\td}[1]{\tilde{#1}}
\newcommand{\sbf}[1]{\scriptsize{\textbf{#1}}}
\newcommand{\stxt}[1]{\scriptsize{\textrm{#1}}}
\newcommand{\suml}[2]{\sum\limits_{#1}^{#2}}
\newcommand{\sumlk}{\sum\limits_{k=0}^{K-1}}
\newcommand{\eqhsp}{\hspace{10 pt}}
\newcommand{\tx}[1]{\texttt{#1}}
\newcommand{\Hz}{\ \tx{Hz}}
\newcommand{\sinc}{\tx{sinc}}
\newcommand{\diag}{\mathrm{diag}}
\newcommand{\MAI}{\tx{MAI}}
\newcommand{\ISI}{\tx{ISI}}
\newcommand{\IBI}{\tx{IBI}}
\newcommand{\CN}{\tx{CN}}
\newcommand{\CP}{\tx{CP}}
\newcommand{\ZP}{\tx{ZP}}
\newcommand{\ZF}{\tx{ZF}}
\newcommand{\SP}{\tx{SP}}
\newcommand{\MMSE}{\tx{MMSE}}
\newcommand{\MINF}{\tx{MINF}}
\newcommand{\RC}{\tx{MP}}
\newcommand{\MBER}{\tx{MBER}}
\newcommand{\MSNR}{\tx{MSNR}}
\newcommand{\MCAP}{\tx{MCAP}}
\newcommand{\vol}{\tx{vol}}
\newcommand{\ah}{\hat{g}}
\newcommand{\tg}{\tilde{g}}
\newcommand{\teta}{\tilde{\eta}}
\newcommand{\heta}{\hat{\eta}}
\newcommand{\uh}{\m{\hat{s}}}
\newcommand{\eh}{\m{\hat{\eta}}}
\newcommand{\hv}{\m{h}}
\newcommand{\hh}{\m{\hat{h}}}
\newcommand{\Po}{P_{\mathrm{out}}}
\newcommand{\Poh}{\hat{P}_{\mathrm{out}}}
\newcommand{\Ph}{\hat{\gamma}}
\newcommand{\mat}[1]{\begin{matrix}#1\end{matrix}}
\newcommand{\ud}{^{\dagger}}
\newcommand{\C}{\mathcal{C}}
\newcommand{\nn}{\nonumber}
\newcommand{\nInf}{U\rightarrow \infty}

\title{Rydberg Atomic Receivers for Multi-Band Communications and Sensing
}
\author{{Mingyao Cui, Qunsong Zeng, Minze Chen, Zhanwei Wang, Tianqi Mao, Dezhi Zheng, and Kaibin Huang,~\IEEEmembership{Fellow, IEEE}\vspace{-7mm}}
\thanks{M. Cui, Q. Zeng, Z. Wang, and K. Huang are with the Department of Electrical and Electronic Engineering, The University of Hong Kong, Hong Kong SAR, China (Emails: \{mycui, qszeng, zhanweiw\}@eee.hku.hk, huangkb@hku.hk). M. Chen, T. Mao, and D. Zheng are with the Beijing Institute of
 Technology, Beijing, China (Emails: \{chenmz22,  maotq, zhengdezhi\}@bit.edu.cn). Corresponding authors: Q. Zeng; K. Huang.}}
\maketitle

\begin{abstract}
Harnessing multi-level electron transitions, Rydberg Atomic REceivers (RAREs) can detect wireless signals across a wide range of frequency bands, from Megahertz to Terahertz. This capability enables multi-band wireless communications and sensing (CommunSense). Existing research on multi-band RAREs primarily focuses on experimental demonstrations, lacking a tractable model to mathematically characterize their mechanisms. This issue leaves the multi-band RARE as a black box and poses challenges in its practical applications. 
To fill in this gap, this paper investigates the underlying mechanism of multi-band RAREs and explores their optimal performance. 
For the first time, an analytical transfer function with a closed-form expression for multi-band RAREs is derived by solving the quantum response of Rydberg atoms. 
It shows that a multi-band RARE simultaneously serves as a \emph{multi-band atomic mixer} for down-converting multi-band signals and a \emph{multi-band atomic amplifier} that reflects its sensitivity to each band. 
Further analysis of the atomic amplifier unveils that the intrinsic gain at each frequency band can be decoupled into a \emph{global gain} term and a \emph{Rabi attention} term. The former determines the overall sensitivity of a RARE to all frequency bands of wireless signals. The latter influences the allocation of the overall sensitivity to each frequency band, representing a unique attention mechanism of multi-band RAREs. The optimal design of the global gain is provided to maximize the overall sensitivity of multi-band RAREs. Subsequently, the optimal Rabi attentions are also derived to maximize the practical multi-band CommunSense performance. An experiment platform is built to validate the effectiveness of the derived transfer function, and numerical results confirm the superiority of multi-band RAREs.
\end{abstract}

\begin{IEEEkeywords}
Rydberg atomic receivers, wireless communications, wireless sensing, multi-band signal detection. 
\end{IEEEkeywords}


\section{Introduction}
The precise measurement of radio-frequency (RF) signals is fundamental to the digital age, serving as a core operation in wireless communications, remote sensing, e-health, and radar systems.
Originating from the domain of quantum sensing, the Rydberg Atomic REceiver (RARE) has emerged as a new concept in high-precision RF detection by exploiting the quantum properties of Rydberg atoms~\cite{RydMag_Liu2023, RydbMag_Cui2024, Shawei2025, RydMag_Fancher2021,RydChen_Gong2025,gao_rydberg_2025}. 
Specifically, Rydberg atoms are highly excited atoms wherein one or more electrons have transitioned from their ground-state energy level to {an excited} energy state. 
Due to their large transition dipole moments, Rydberg atoms can strongly interact with incident RF signals, triggering electron transitions between resonant energy levels~\cite{RydMag_Fancher2021}.
Capitalizing on these transitions, RAREs can capture the amplitude, frequency, phase, and polarization of RF signals with unparalleled precision.
Consequently, RAREs have the potential to {complement} or even replace traditional RF receivers in the next-generation wireless communications and sensing (CommunSense) systems~\cite{ zhang_rydberg_2024, QuanSense_Zhang2023, AtomicMIMO_Cui2025,QWC_Cui2025, CE_Ryd2025,QWC_kim2025,Rydberg_Gong2025,liu_deep_2022}.

RAREs offer two major advantages over classical receivers: extremely high sensitivity and a wide range of detectable frequency bands. 
{Classical receivers, comprising metallic antennas and RF front-end circuitry, are fundamentally constrained by Johnson-Nyquist noise that arises from the random thermal motion of free electrons~\cite{nyquist_thermal_1928}.}
In contrast, the {quantum process governing} RAREs, electron transition {between Rydberg states}, is {inherently} immune to thermal noise.  
This {immunity enables} RAREs {to achieve} exceptional sensitivity to RF signals, {\color{black} potentially 1$\sim$2 orders of magnitude higher than that of traditional counterparts~\cite{RydModel_Gong2024, RydGong_2025}.}
Additionally, {conventional metallic antennas require physical dimensions comparable to the carrier wavelength to efficiently couple with incident RF signals.}
{Their detectable frequency ranges are thus limited by their geometric sizes~\cite{chu_physical_1948,Multiband_Abo2024}.}
In contrast, by utilizing electron transitions between numerous energy levels, a single RARE can simultaneously detect multiple distinct frequency bands, ranging from Megahertz (MHz) to Terahertz (THz), without any change of the hardware\cite{zhang_rydberg_2024}. 
This {distinctive} feature underpins the development {of} integrated full-frequency wireless CommunSense platforms, {unifying} sub-6G, midband, millimeter wave (mmWave), and THz bands {into a single receiver device}. 
This paper exploits the RAREs' multi-band detection capability to reshape wireless CommunSense systems.

Existing studies have {experimentally validated} the concurrent multi-band CommunSense capabilities of RAREs. 
The {initial demonstration} of dual-band communications was reported in \cite{MBRARE_Holloway2021}, {where} two species of Rydberg alkali atoms, rubidium and cesium, each with distinct energy-level structures, were mixed in a single vapor cell. 
This configuration enabled the simultaneous detection of amplitude-modulated (AM) signals at 19.626 GHz and frequency-modulated (FM) signals at 20.644 GHz. 
{Subsequent research} demonstrated that leveraging the electron transitions between different energy levels within a single atom species, (e.g., cesium), is sufficient to capture multiple frequency bands, which has become a common practice~\cite{RydMultiband_Du2022}. 
In particular,  the concurrent measurement of frequency bands at 1.72, 12.11, 27.42, 65.11, and 115.75 GHz has been reported in \cite{Rydmultiband_Meyer2023}. It utilizes the technology of \emph{quantum heterodyne sensing}, which introduces external reference signals for assisting signal detection, to excite electron transitions between chosen energy levels, $56D_{5/2}$, $57P_{3/2}$, $54F_{7/2}$, $52F_{7/2}$, and $59P_{3/2}$. 
{Further advances include chip-integrated RAREs for dual-band operation at 300 MHz and 24 GHz~\cite{DBRARE_Ding2024}, and a cascaded orbital angular-momentum transition to enable coverage from 128 MHz to 0.61 THz~\cite{Rydmultiband_Allinson2024}.}
Additionally, experiments have also observed a sensitivity tradeoff of the multi-band RARE, the exclusive effect~\cite{photonics10030328}. It states that a RARE's sensitivities to different frequency bands can influence each other, an increase in sensitivity at one band might result in decreased sensitivity at other bands. 
Some recent studies have also explored the applications of multi-band RAREs, including frequency-hopping communication and angles-of-arrival (AoAs) sensing~\cite{Frequencyhop_Wen2024, MBRARE_Kim2025}. 

RARE enabled multi-band CommunSense is still in its infancy. 
Currently, there exists no tractable model to explicitly characterize the transfer function of multi-band RAREs. 
{Prior studies largely in} the physics community have {predominantly emphasized} experimental verifications~\cite{MBRARE_Holloway2021, RydMultiband_Du2022, Rydmultiband_Meyer2023, DBRARE_Ding2024, Rydmultiband_Allinson2024, photonics10030328, Frequencyhop_Wen2024}. 
{In the existing literature,} only the work in  \cite{MBRARE_Kim2025} makes an initial attempt to establish a signal model for multi-band RAREs. 
However, 
{\color{black} 
the derivation directly adopts the conventional single-band RARE framework, which might be over-simplified as it neglects some delicate but important characteristics of multi-band RAREs.
}
{The lack of a tractable theory renders} the multi-band RARE {as} a black box, 
{hindering rigorous analysis of their quantum phenomena (e.g., the mentioned sensitivity tradeoff) and performance optimization.}

To fill in this knowledge gap, we make an attempt to theoretically characterize the underlying mechanisms of multi-band RAREs and apply the results to optimize their performance in CommunSense systems. 
{We leverage multi-level electron transitions in Rydberg atoms to enable simultaneous multi-band signal reception, augmented by quantum heterodyne sensing to facilitate signal detection.}
It is worth emphasizing that the considered multi-band RARE architecture enables detection of disjoint bands while the conventional orthogonal frequency division multiplexing (OFDM) systems are confined to a single contiguous frequency block (e.g., a 100 MHz channel at 3.5 GHz in 5G standard~\cite{5GNR_Park2017}). 
The key contributions and findings of our work are summarized as follows.
\begin{itemize}
    \item \textbf{Theoretical framework of multi-band RAREs:} 
We develop a theoretical framework to characterize the operational mechanisms underlying multi-band RAREs. {\color{black} By analytically solving the Lindblad master equation that governs electron transition, we derive for the first time a closed‑form transfer function for multi‑band RAREs.}
This model reveals their dual functionality as a \emph{multi-band atomic mixer} and a \emph{multi-band atomic amplifier}. Specifically, each RF signal component undergoes individual down-conversion to the intermediate frequency (IF) through its resonant energy‑level transition. The resulting IF component is amplified by a band‑specific gain factor. Our analysis demonstrates that this gain factor directly determines the receiver's sensitivity per band.
    \item \textbf{Sensitivity {maximization} of multi-band RAREs:} 
    Building on this framework, we maximize the sensitivity of a RARE across all frequency bands by optimizing the Rabi frequencies (electron transition strengths) of the external reference signals used in quantum heterodyne sensing. This problem is intractable due to the mutual coupling among multi-band sensitivities. To solve it, we prove that the band-specific gain factor can be decomposed into a global gain term (common to all bands) and a band‑specific Rabi attention term.
    The global gain determines the overall sensitivity to all bands, which depends solely on the root-sum-square of all Rabi frequencies, termed the \emph{effective Rabi frequency}. 
    The maximization of the RARE's overall sensitivity is thus achievable by optimizing this effective Rabi frequency only, for which the closed‑form global optimum is derived. 
    The Rabi attentions characterize the distribution of Rabi frequencies over all frequency bands, which govern the allocation of overall sensitivity to each band.
    They represent a unique attention mechanism of multi-band RAREs, providing additional degrees of freedom for optimization according to a downstream CommunSense metric. 
    \item \textbf{CommunSense applications of multi-band RAREs:} 
    Last, we apply the above results to optimize the performance of systems supporting multi-band communications and multi-granularity sensing. 
    {Their performances are evaluated using the metrics of spectral efficiency (SE) for communications and normalized Cramér-Rao lower bound (NCRLB) for sensing, both averaged across frequency bands.}
    {We derive optimal Rabi attention allocations for each scenario.}
    Numerical experiments confirm that the optimized multi‑band RARE significantly outperforms both classical receivers and existing single‑band RARE designs in these CommunSense applications.
\end{itemize}

The remainder of the paper is organized as follows. Section~\ref{sec:sys} presents the system model.  The theoretical foundation of the multi-band RARE is established in Section~\ref{sec:2}. The maximal sensitivity of a multi-band RARE is derived in Section~\ref{sec:3}. Section~\ref{sec:4} elaborates on the applications of multi-band RARE in CommunSense systems and the associated optimal Rabi attention allocation.
Experiments are conducted in Section~\ref{sec:6}, and conclusions are drawn in Section~\ref{sec:7}.

\begin{figure}[t!]
	\centering
	\subfigure[]
	{\includegraphics[width=3in]{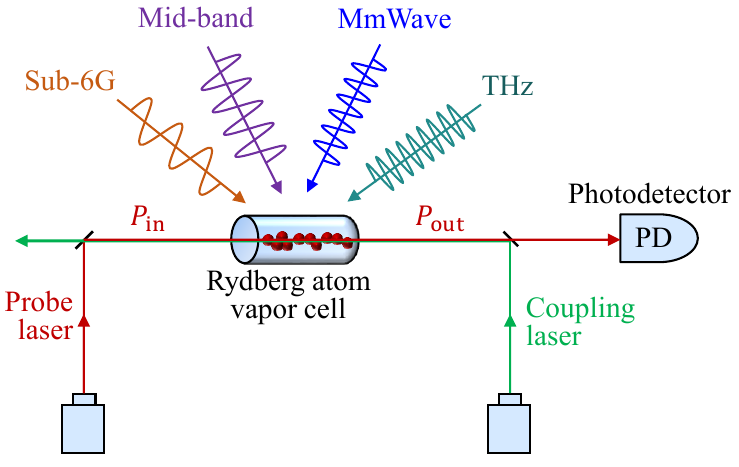}}
	\subfigure[]{\includegraphics[width=3in]{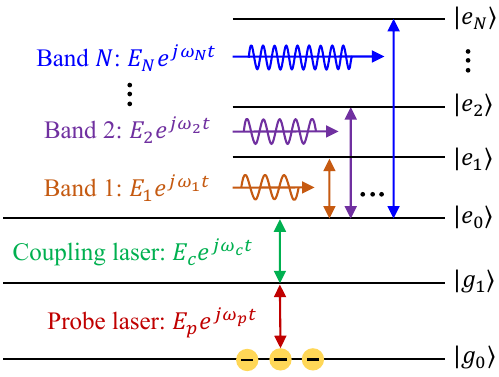}} 
	\caption{(a) Architecture of a multi-band RARE. (b) Multi-level electron transitions for measuring multi-band signals.} %
	\vspace*{-1em}
	\label{img:RARE-MSAC}
\end{figure}

\section{System Model}\label{sec:sys}

Consider the multi-band RARE system illustrated in Fig.~\ref{img:RARE-MSAC} (a). Multi-band signals spanning the frequency range from sub-6G to THz are simultaneously detected to offer diverse CommunSense services. 
For {classical} receivers, measuring these signals necessitates multiple antennas and {RF} front-end circuits due to the wavelength-dependent antenna size~\cite{Multiband_Sim2020,Multiband_Abo2024}.  However, a single RARE is sufficient to capture them by exploiting the multi-level electron transitions as depicted in Fig.~\ref{img:RARE-MSAC} (b). Relevant models are described in this section.

\subsection{Incident Multi-Band Signal Model}
The incident multi-band signal consists of $N$ disjoint frequency bands: 
\begin{align}\label{eq:multi-band_signal}
    E_{\rm RF}(t) = \sum_{n = 1}^N E_{{\rm RF}, n}(t). 
\end{align}
Here, $E_{{\rm RF}, n}(t)$ represents the incident RF signal at band $n$:
\begin{align}\label{eq:RF}
    E_{{\rm RF},n}(t) = \Re{E_n(t)e^{j\omega_n t}},
\end{align}
where $\omega_n$ denotes the angular frequency of band $n$ and $E_n(t)\in\mathbb{C}$ the baseband component. The occupied bandwidth of each RF signal $E_{{\rm RF},n}(t)$ is denoted as $B_n$. 

Existing RARE systems popularly employ the quantum heterodyne sensing technique to facilitate signal detection~\cite{Rydmultiband_Meyer2023}. 
{This technique introduces external local oscillators (LOs), also known as reference sources, to generate reference signals that match the target wireless frequencies, $\omega_n,\:\forall n$. 
The composite baseband signal, $E_n(t)$, thereby consists of two parts: the LO generated reference signal, denoted as $E_{{\rm r},n} \in \mathbb{C}$, and the data signal to be detected, denoted as $E_{{\rm s},n} \in \mathbb{C}$. Besides, to enable phase detection, a slight frequency offset, $\delta_n$, between the reference and data signals is usually introduced, serving as the IF component. Without loss of generality, the IFs are configured in an ascending order $\delta_1 < \delta_2<\cdots<\delta_N$. 
As a result, the composite baseband signal $E_n(t) $ can be written as 
\begin{align}
    E_n(t) = E_{{\rm r},n} + E_{{\rm s},n}e^{j\delta_n t}. 
\end{align}
In practice, the reference signal $E_{{\rm r},n}$ is known to the RARE. The CommunSense information is embedded within the data signal component $E_{{\rm s},n}$, as elaborated below.

\begin{itemize}
    \item \emph{Communication signal:} Quadrature amplitude modulation (QAM) symbols are transmitted from $N$ heterogeneous devices operating at distinct frequency bands. Let $s_n\in\mathbb{S}_n$ be the QAM symbol from the $n$-th device, normalized such that $\mathbb{E}[|s_n|^2]=1$, where $\mathbb{S}_n$ denotes the QAM constellation set. 
    Then, the baseband data signal, $E_{{\rm s},n}$, for communication service is modeled as\footnote{\color{black}
    The factor, $2\pi$, in \eqref{eq:sn} arises from the relationship between the time-averaged power flux density $S$ and the peak amplitude of the oscillating electric field $E$. For an isotropic radiator, $S = \frac{P}{4\pi r^2}$, and for a plane wave, $S = \frac{E^2}{2Z_0}$, which  together yield $E\propto \sqrt{\frac{Z_0P}{2\pi r^2}}$. 
    } 
\begin{align}\label{eq:sn}
    E_{{\rm s},n} = \sqrt{\frac{Z_0P_{{\rm s},n}}{2\pi r^2_n}}e^{-j(\omega_n+\delta_n)\tau_n}s_n = {h}_ns_n,
\end{align}
where $Z_0$ is the vacuum impedance, $\sqrt{P_{{\rm s},n}}$ the transmit power, $r_n$ the link distance, $\tau_n$ the propagation delay, and ${h}_n = \sqrt{\frac{Z_0P_{{\rm s},n}}{2\pi r^2_n}}e^{-j(\omega_n+\delta_n)\tau_n}$ the channel coefficient. 

\item \emph{Sensing signal:} 
As frequency $\omega$ increases, a minute target displacement $d$ will induce a significant phase shift $-\frac{\omega d}{c}$ to the RF signal, with $c$ the speed of light. By simultaneously detecting $N$ frequency bands, multi-band RAREs achieve multi-scale sensing resolution from coarse (e.g., hand tracking) to fine (e.g., finger tracking). Specifically, at band $n$, the sensing target  undergoes a slight displacement, $d_n$, which introduces an additional phase, $-\frac{(\omega_n+\delta_n)}{c}d_n$, to the sensing signal. Thus, $E_{{\rm s},n}$ is formulated as 
\begin{align}\label{eq:dn}
    E_{{\rm s},n} = {h}_ne^{-j\frac{(\omega_n+\delta_n)}{c}d_n}\overset{(a)}{\approx} {h}_ne^{-j\frac{\omega_n}{c}d_n},
\end{align}
{where} the approximation (a) holds because $\omega_n \gg \delta_n$. 
{To avoid phase ambiguity, displacements are constrained to $d_n\in\left[-\frac{\pi c}{\omega_n},\frac{\pi c}{\omega_n}\right]$~\cite{QWC_Cui2025, QuanSense_Zhang2023}.}
\end{itemize}}
 



\subsection{Multi-Level Electron Transition Model}
As depicted in Fig.~\ref{img:RARE-MSAC}, a RARE utilizes the multi-level electron transitions of Rydberg atoms to interact with the multi-band signal, $E_{\rm RF}(t)$. This interaction changes the quantum {states} of Rydberg atoms, resulting in variations in the absorption rate of the atomic ensemble to an incoming laser. By propagating a probe laser through these atoms and then monitoring the output probe-laser power using a photodetector (PD), the quantum state can be read out. This facilitates the detection of multi-band signals, $\{E_{{\rm s},n}\}$, and the recovery of multi-band CommunSense information, $\{s_n\}$ and $\{d_n\}$. The above physical processes are elaborated below. 

\subsubsection{Quantum state} In this system, a probe laser and a coupling laser propagate in opposite directions through a vapor cell filled with alkali atoms to prepare Rydberg atoms. 
As shown in Fig.~\ref{img:RARE-MSAC}(b), the probe laser excites electrons to jump from the ground state, denoted as $\ket{g_0}$, to a lowly excited state, denoted as $\ket{g_1}$. The coupling laser further triggers the electron transition from $\ket{g_1}$ to an initial Rydberg state, $\ket{e_0}$, thereby creating Rydberg atoms. These atoms interact with each frequency band of the incident signal, say band $n$, via the electron transition $\ket{e_0}\rightarrow \ket{e_n}$, $n\in\{1,2,\cdots,N\}$, respectively. 
As a result, the quantum response of each Rydberg atom is characterized by an $(N+3)$-level system. The associated quantum state is thus an $(N+3)\times(N+3)$ Hermitian matrix
    \begin{align}\label{eq:density}
    \boldsymbol{\rho} = \left[
    \begin{array}{cccc}
         \rho_{11}& \rho_{12} & \cdots & \rho_{1,N+3} \\
         \rho_{21}& \rho_{22} & \cdots & \rho_{2, N+3}\\
         \vdots & \vdots &\ddots & \vdots \\
         \rho_{N+3,1} & \rho_{N+3,2} & \cdots & \rho_{N+3,N+3}
    \end{array}
    \right],
\end{align}
satisfying $\Tr{\boldsymbol{\rho}} = 1$ and $\boldsymbol{\rho}\succeq 0$. 

\subsubsection{Hamiltonian operator} The Hamiltonian operator governs electron transitions. 
Specifically, we define $\mu_{g_0g_1}$, $\mu_{g_1e_0}$, and $\mu_{n}$, $\forall n\in\{1,\cdots,N\}$ as the transition dipole moments associated with the electron transitions, $\ket{g_0}\rightarrow \ket{g_1}$, $\ket{g_1}\rightarrow \ket{e_0}$, and $\ket{e_0}\rightarrow \ket{e_n}$, $\forall n$, respectively. Let $\omega_{g_0g_1}$, $\omega_{g_1e_0}$, and $\omega_{e_0e_n}$, $\forall n$ be the transition frequencies of these electron transitions. 
We suppose that the {\color{black}amplitude and frequency} of probe laser (coupling laser) are fixed as $E_{\rm p}$ ($E_{\rm c}$) and $\omega_{\rm p}$ ($\omega_{\rm c}$), respectively. The resonant cases for all electron transitions are considered, meaning that $\omega_{\rm p} = \omega_{g_0g_1}$, $\omega_{\rm c} = \omega_{g_1e_0}$, and $\omega_n = \omega_{e_0e_n},\forall n$. 
Following these definitions, the Rabi frequencies of the probe and coupling lasers, which quantify the strength of the interaction between the signal and the atom, can be represented as~\cite{AtomicPhysics}
\begin{align}
         \Omega_{\rm p} = \frac{\mu_{g_0 g_1}}{\hbar}\left|E_{\rm p}\right|  \quad{\rm and}\quad
    \Omega_{\rm c} = \frac{\mu_{g_1 e_0}}{\hbar} \left|E_{\rm c}\right|,
\end{align}
respectively, where $\hbar$ is the reduced Plank constant.
Furthermore, the Rabi frequency pertaining to each RF signal component, $E_{{\rm RF}, n}(t), \forall n$, is given as 
\begin{align}\label{eq:Rabi}
    \Omega_n &= \frac{\mu_{n}}{\hbar}\left|E_n(t)\right| = \frac{\mu_n}{\hbar}|E_{{\rm r},n} + E_{{\rm s},n}e^{j\delta_n t}|, \notag\\ 
    &\overset{\Delta}{=}|\Omega_{{\rm r},n} + \Omega_{{\rm s},n}e^{j(\delta_n t+\phi_n)}|. 
\end{align}
The variable ($\phi_n = \angle E_{{\rm s},n} - \angle E_{{\rm r},n}$) denotes the phase difference between $E_{{\rm s},n}$ and $E_{{\rm r},n}$. Besides,  $\Omega_{{\rm r},n} \overset{\Delta}{=}\frac{\mu_n}{\hbar}|E_{{\rm r},n}|$ and $\Omega_{{\rm s},n}\overset{\Delta}{=}\frac{\mu_n}{\hbar}|E_{{\rm s},n}|$ refer to the Rabi frequencies of the reference and data signal components, respectively. For convenience in future use, we define the Rabi frequency vectors of the reference {signals}, data {signals}, and composite signals as $\boldsymbol{\Omega}_r = [\Omega_{r,1},\cdots,\Omega_{r,N}]^T$, $\boldsymbol{\Omega}_s = [\Omega_{s,1},\cdots,\Omega_{s,N}]^T$, and $\boldsymbol{\Omega} = [\Omega_{1},\cdots,\Omega_{N}]^T$.  {\color{black}  }
As a result, the Hamiltonian operator for this $(N+3)$-level quantum system is represented as  
\begin{align}\label{eq:hamiltonian}
    {\mb{H}} = \frac{\hbar}{2}\left[
    \begin{array}{ccccccc}
         0 &\Omega_{\rm p} &0 &0 &\cdots & 0  \\
         \Omega_{\rm p} &0 &\Omega_{\rm c} &0 &\cdots & 0 \\
         0 &\Omega_{\rm c} &0 &\Omega_1 &\cdots & \Omega_N\\
         0 &0 &\Omega_1 &0 &\cdots & 0\\
         \vdots &\vdots &\vdots &\vdots &\ddots &\vdots\\
         0 &0 &\Omega_N &0 &\cdots & 0\\
    \end{array}
    \right],
\end{align}
which encompasses the Rabi frequencies of all external signals, $\Omega_{\rm p}$, $\Omega_{\rm c}$, and $\Omega_n, \forall n$.

\subsubsection{Lindblad master equation}
The Hamiltonian operator, ${\mb{H}}$, drives the quantum state, $\boldsymbol{\rho}$, to evolve over time, during which the Rabi frequencies are encoded into the quantum state. To be specific, the evolution follows the Lindblad master equation~\cite{RydNP_Jing2020}
    \begin{align}\label{eq:lindblad}
    \frac{\partial \boldsymbol{\rho}}{\partial t} = \frac{j}{\hbar}[\boldsymbol{\rho},\mb{H}] + \mathcal{L},
\end{align}
where $[\boldsymbol{\rho}, \mb{H}] \overset{\Delta}{=}  \boldsymbol{\rho} \mb{H} - \mb{H}\boldsymbol{\rho}$ represents the commutator between $\boldsymbol{\rho}$ and $\mb{H}$.  
The operator $\mathcal{L}$ characterizes the noncoherent relaxation of the system,  which is given in \eqref{eq:decay}.
\begin{figure*}[ht]
\centering
\begin{align}\label{eq:decay}
    \mathcal{L} = \left[
    \begin{array}{ccccc}
         \gamma_2\rho_{22} + \sum_{i = 4}^{N+3}\gamma_i \rho_{ii}& -\frac{\gamma_2}{2}\rho_{12} & -\frac{\gamma_3}{2}\rho_{13} & \cdots  & -\frac{\gamma_{N+3}}{2}\rho_{1,N+3} \\
         -\frac{\gamma_2 }{2}\rho_{21}& 
         \gamma_3\rho_{33} - \gamma_2\rho_{22}
         & -\frac{\gamma_{23}}{2}\rho_{23} &
         \cdots& -\frac{\gamma_{2,N+3}}{2}\rho_{2,N+3}\\
         -\frac{\gamma_3 }{2}\rho_{31} &
         -\frac{\gamma_{23}}{2}\rho_{32} &
         -\gamma_3\rho_{33} &
         \cdots& -\frac{\gamma_{3,N+3}}{2}\rho_{3,N+3} \\
         \vdots & \vdots & \vdots & \ddots &  \vdots\\
      -\frac{\gamma_{N+3}}{2}\rho_{N+3,1} & -\frac{\gamma_{2,{N+3}}}{2}\rho_{{N+3}, 2} & - \frac{\gamma_{3,N+3}}{2}  \rho_{N+3, 3} & \cdots &-\gamma_{N+3}\rho_{N+3, N+3}
    \end{array}\right].
\end{align}
\hrulefill
\end{figure*}
Here, $\gamma_{ij} = (\gamma_i + \gamma_j)/2$, where $\gamma_i\: (i \in\{2,3,\cdots,N+3\})$ refers to the decay rate of each energy level. In practice, the decay rate $\gamma_2$ is greater than  $\gamma_i\: (i \in\{3,\cdots,N+3\})$ by orders of magnitude~\cite{RydNP_Jing2020}. Thus, \eqref{eq:decay} is usually simplified to 
\begin{align}
    \mathcal{L} = -\frac{\gamma_2}{2}\left[
    \begin{array}{cccccc}
         -2\rho_{22}& \rho_{12} & 0 &\cdots & 0 \\
         \rho_{21}& 
          2 \rho_{22}
         & \rho_{23}  &
         \cdots& \rho_{2,N+3}\\
         0 &
         \rho_{32} &0 &\cdots&0 \\
         \vdots&\vdots&\vdots&\ddots&\vdots\\
      0& \rho_{N+3,2} & 0  &\cdots& 0
    \end{array}
    \right], \notag
\end{align}
without loss of accuracy~\cite{RydNP_Jing2020,liu_continuous-frequency_2022}.

\subsubsection{Measured quantum state}
The evolved quantum states are measured by a PD, which monitors the power loss of the probe laser when passing through the vapor cell. 
Specifically, let $P_{\rm in}$ be the input power of probe laser.
According to the adiabatic approximation, the output power of probe laser exponentially decays with the steady-state solution of the coherence term $\rho_{12}$ (the $(1, 2)$-th entry of $\boldsymbol{\rho}$)~\cite{RydNP_Jing2020}:  \begin{align}\label{eq:laserpower}
    \mathcal{P}(\boldsymbol{\Omega}) = P_{\rm in}\exp(-{C_0\Im{\rho_{12}}}),
\end{align} 
where $\color{black} C_0 \overset{\Delta}{=} \frac{2N_0\mu_{g_0g_1}^2k_pL}{\epsilon_0\hbar\Omega_{\rm p}}$. Here, $N_0$ represents the atomic density, $\epsilon_0$ the vacuum permittivity, $L$ the length of vapor cell, and $k_p$ the wavenumber of probe laser. 

The PD captures the probe laser and converts it into photocurrent:
\begin{align}\label{eq:current}
    I(\mb{\Omega}) = \frac{q\eta}{\hbar\omega_{\rm p}}   \mathcal{P}(\boldsymbol{\Omega}) =  I_{\rm in}\exp(-{C_0\Im{\rho_{12}}}),
\end{align}
Here, $\eta$ is the quantum efficiency of the PD, and $q$ the elementary charge.  The equivalent input photocurrent, $I_{\rm in} \triangleq \frac{q\eta P_{\rm in}}{\hbar\omega_{\rm p}}$, is introduced to characterize the photocurrent induced by the input power of probe laser $P_{\rm in}$.  

\subsection{An Open Problem}
Equations \eqref{eq:Rabi}$\sim$\eqref{eq:current} characterize the microwave-to-optical conversion process of a multi-band RARE: 
\begin{align}\label{eq:mapping}
    \{E_{{\rm s},n}\}_{n = 1}^N \overset{\eqref{eq:Rabi}}{\rightarrow} \boldsymbol{\Omega} \overset{\eqref{eq:lindblad}}{\rightarrow} \rho_{12} \overset{\eqref{eq:laserpower}}{\rightarrow} \mathcal{P}({\boldsymbol{{\Omega}}}) \overset{\eqref{eq:current}}{\rightarrow} I(\mb{\Omega}).
\end{align}
The multi-band RF signals, $\{E_{{\rm s},n}\}$, are initially mapped onto the Rabi frequency vector $\boldsymbol{\Omega}$, which in turn influences the quantum coherence term $\rho_{12}$ and is finally converted into the  laser power, $\mathcal{P}(\boldsymbol{\Omega})$, and photocurrent, $I(\mb{\Omega})$. This cascade enables the simultaneous detection of multiple RF bands from the measured photocurrent. However, obtaining an explicit mapping between $\rho_{12}$ and $\boldsymbol{\Omega}$ requires solving the steady-state solution of the Lindblad master equation \eqref{eq:lindblad}. 
 Although analytical solutions for single‑band RAREs ($N=1$) are well established~\cite{RydNP_Jing2020,liu_continuous-frequency_2022}, no such closed‑form results are available for multi‑band systems with $N > 1$ in the literature. Instead, existing efforts on multi-band RARE~\cite{MBRARE_Holloway2021, RydMultiband_Du2022, Rydmultiband_Meyer2023, DBRARE_Ding2024, Rydmultiband_Allinson2024, photonics10030328, Frequencyhop_Wen2024} rely primarily on experimental or numerical approaches to obtain the quantum state $\rho_{12}$ without solving the Lindblad master equation \eqref{eq:lindblad} analytically.
As a result, the explicit transfer function from $\{E_{s,n}\}_{n = 1}^N$ to $I({\boldsymbol{{\Omega}}})$ remains unknown. Deriving an analytical solution to \eqref{eq:lindblad} is therefore the theme of the following section.


\section{Theoretical Framework of Multi-Band RARE}\label{sec:2}
This section presents a theoretical framework of multi-band RAREs. 
We first derive the steady-state solution of the master equation in \eqref{eq:lindblad} for a multi‑level electron‑transition system.
This solution then serves as the foundation for obtaining the explicit transfer function of a multi‑band RARE.

\subsection{Solving the Lindblad Master Equation}

We focus on the steady-state solution of the coherence term $\rho_{12}$, 
due to its direct relationship with the probe‑laser transmission. The following theorem provides its closed‑form expression.
{\color{black}
\begin{theorem}\rm
    The steady-state solution of $\rho_{12}$ to the Lindblad master equation in \eqref{eq:lindblad} is given by 
\begin{align}\label{eq:state}
    \boxed{\rho_{12} = j\frac{  \rho_0\sum_{n = 1}^N\Omega_n^2}{\sum_{n = 1}^N\Omega_n^2 + \Gamma^2},}
\end{align}
where $\rho_0 \overset{\Delta}{=} \frac{\gamma_2\Omega_{\rm p}}{\gamma_2^2 + 2\Omega_{\rm p}^2}$ and $\Gamma \overset{\Delta}{=} \sqrt{\frac{2\Omega_{\rm p}^2(\Omega_{\rm c}^2 + \Omega_{\rm p}^2)}{\gamma_2^2 + 2\Omega_{\rm p}^2}}$ is related to the electromagnetically induced transparency (EIT) linewidth. 
\end{theorem}
\begin{proof}
    {The expression is derived by imposing the steady-state condition $\frac{\partial \boldsymbol{\rho}}{\partial t} = 0$ and solving the resulting linear system obtained from the master equation \eqref{eq:lindblad}, which comprises $(N+3)^2$ variables corresponding to all elements of the density matrix. }
    {For conciseness, the full derivation is omitted here; though the complete solution for all elements of $\boldsymbol{\rho}$, including $\rho_{12}$, is rigorously validated and provided in Appendix~\ref{app:solution} for reproducibility.}
\end{proof}
}

{Theorem 1 demonstrates that the coherence term $\rho_{12}$  depends explicitly on the sum square of Rabi frequencies $\sum_{n = 1}^N\Omega_n^2$, reflecting the collective contributions from all $N$ frequency bands to the Rydberg state. This collective dependence enables the RARE to extract information concurrently from multiple bands.}
{Moreover, the derived quantum coherence generalizes the well‑known single‑band result~\cite{RydNP_Jing2020,liu_continuous-frequency_2022}}
\begin{align}\label{eq:rho21_2}
    \rho_{12} = j\frac{ \rho_0 \Omega_1^2}{\Omega_1^2 + \Gamma^2}.
\end{align}  
{Notably, setting only one Rabi frequency non-zero ($\Omega_n \neq0$, $\Omega_{n'}=0$, $\forall n'\neq n$) in Theorem 1 exactly recovers the single-band solution in \eqref{eq:rho21_2}. This validates the consistency of our multi‑band generalization.}


\subsection{Transfer Function of Multi-band RAREs}
\subsubsection{Transfer function at the IF band}
Building on the derived steady-state solution, we now proceed to derive the transfer function of a multi‑band RARE by following the signal‑mapping flow outlined in \eqref{eq:mapping}. Substituting \eqref{eq:state} into \eqref{eq:laserpower} yields the explicit expression for the measured photocurrent
\begin{align}\label{eq:y1}
    y(t) &= {I}(\boldsymbol{\Omega}) = I_{\rm in}\exp\left(-\frac{ \chi_0 \sum_{n=1}^N\Omega_n^2}{\sum_{n=1}^N\Omega_n^2 + \Gamma^2}\right),\\&= I_{\rm in}\exp\left(-\frac{  \chi_0\sum_{n=1}^N|\Omega_{{\rm r},n} + \Omega_{{\rm s},n}e^{j(\delta_nt + \phi_n)}|^2}{\sum_{n=1}^N|\Omega_{{\rm r},n} + \Omega_{{\rm s},n}e^{j(\delta_nt + \phi_n)}|^2 + \Gamma^2}\right),\notag
\end{align}
where $\chi_0 \overset{\Delta}{=}\rho_0 C_0$ denotes the laser attenuation coefficient. This transfer function unveils the nonlinear relationship between the photocurrent and multi-band signals.  Such nonlinearity makes it challenging to decipher the data signals $E_{{\rm s},n},\forall n$, especially their phase information. To overcome this, the reference sources are typically placed close to the RARE, ensuring that the reference signal is orders of magnitude stronger than the data signal, e.g., $|E_{{\rm r},n}| \gg |E_{{\rm s},n}|$ and $\Omega_{{\rm r},n}\gg \Omega_{{\rm s},n}$~\cite{Rydphase_Holloway2019}. 
In this strong-reference regime, the data signals act as weak perturbations to the quantum system. 
Consequently, the nonlinear transfer function in \eqref{eq:y1} can be accurately linearized via a first-order Taylor expansion, greatly simplifying signal detection. Specifically, the derivative of $I(\boldsymbol{\Omega})$ with {respect} to (w.r.t.) $\Omega_{n},\forall n$ is expressed as: 
\begin{align}\label{eq:gradient}\color{black}
    \frac{\partial I(\boldsymbol{\Omega})}{\partial \Omega_n} = -I(\boldsymbol{\Omega})\frac{2\chi_0\Gamma^2\Omega_n}{(\sum_{m=1}^N\Omega_m^2 + \Gamma^2)^2}. 
\end{align}
Using \eqref{eq:gradient}, the linearized photocurrent can be expressed as shown in \eqref{eq:y2}.
\ifx\onecol\undefined
\begin{figure*}[ht]\color{black}
\centering
    \begin{align} \label{eq:y2}
    y(t) &\approx \underbrace{I(\boldsymbol{\Omega}_{r})}_{\rm DC\:bias} + \sum_{n = 1}^N\underbrace{\frac{\partial I(\boldsymbol{\Omega}_{r})}{\partial \Omega_{{\rm r},n}}\frac{\mu_n}{\hbar}}_{{\rm Gain\:for\:band}\:n} \underbrace{\Re{E_{{\rm s},n}e^{j(\delta_n t - \angle E_{{\rm r},n})}}}_{{\rm Data\:signal\:at\:band\:}n} \overset{\Delta}{=}I_{\rm r} - \sum_{n = 1}^N\kappa_n {\Re{E_{{\rm s},n}e^{j(\delta_n t - \angle E_{{\rm r},n})}}}.
\end{align}
\hrulefill
\end{figure*}
\else 
\begin{align} \label{eq:y2}
    y(t) &\approx \underbrace{\mathcal{P}(\boldsymbol{\Omega}_{r})}_{\rm DC\:bias} + \sum_{n = 1}^N\underbrace{\frac{\partial \mathcal{P}(\boldsymbol{\Omega}_{r})}{\partial \Omega_{{\rm r},n}}\frac{\mu_n}{\hbar}}_{{\rm Gain\:for\:band}\:n} \underbrace{\Re{E_{{\rm s},n}e^{j(\delta_n t - \angle E_{{\rm r},n})}}}_{{\rm Data\:signal\:at\:band\:}n} \overset{\Delta}{=}I_{\rm r} + \sum_{n = 1}^N\kappa_n {\Re{E_{{\rm s},n}e^{j(\delta_n t - \angle E_{{\rm r},n})}}}
\end{align}
\fi 
 It is observed that the resulting photocurrent consists of three {components}: 1) a direct current (DC) bias $I_{\rm r} \overset{\Delta}{=}I(\boldsymbol{\Omega}_{r})$ generated by the strong reference signal; 
 {2) a band-specific amplification coefficient}
 \begin{align}
 \color{black} \kappa_n \overset{\Delta}{=} - \frac{\partial I(\boldsymbol{\Omega}_{r})}{\partial \Omega_{{\rm r},n}}\frac{\mu_n}{\hbar} = I(\boldsymbol{\Omega}_r)\frac{2\chi_0\Gamma^2\mu_n\Omega_{{\rm r},n}}{\hbar(\sum_{m=1}^N\Omega_{r,m}^2 + \Gamma^2)^2},
 \end{align} which quantifies the intrinsic gain of the RARE for the 
$n$-th band; 
 and 
 3) the modulated data signal, $\Re{E_{{\rm s},n}e^{j(\delta_n t - \angle E_{{\rm r},n})}}$, which carries the information payload.
From \eqref{eq:y2}, we can thus identify the physical significance of a multi‑band RARE.

\begin{figure*}
    \centering
    \includegraphics[width=6.5in]{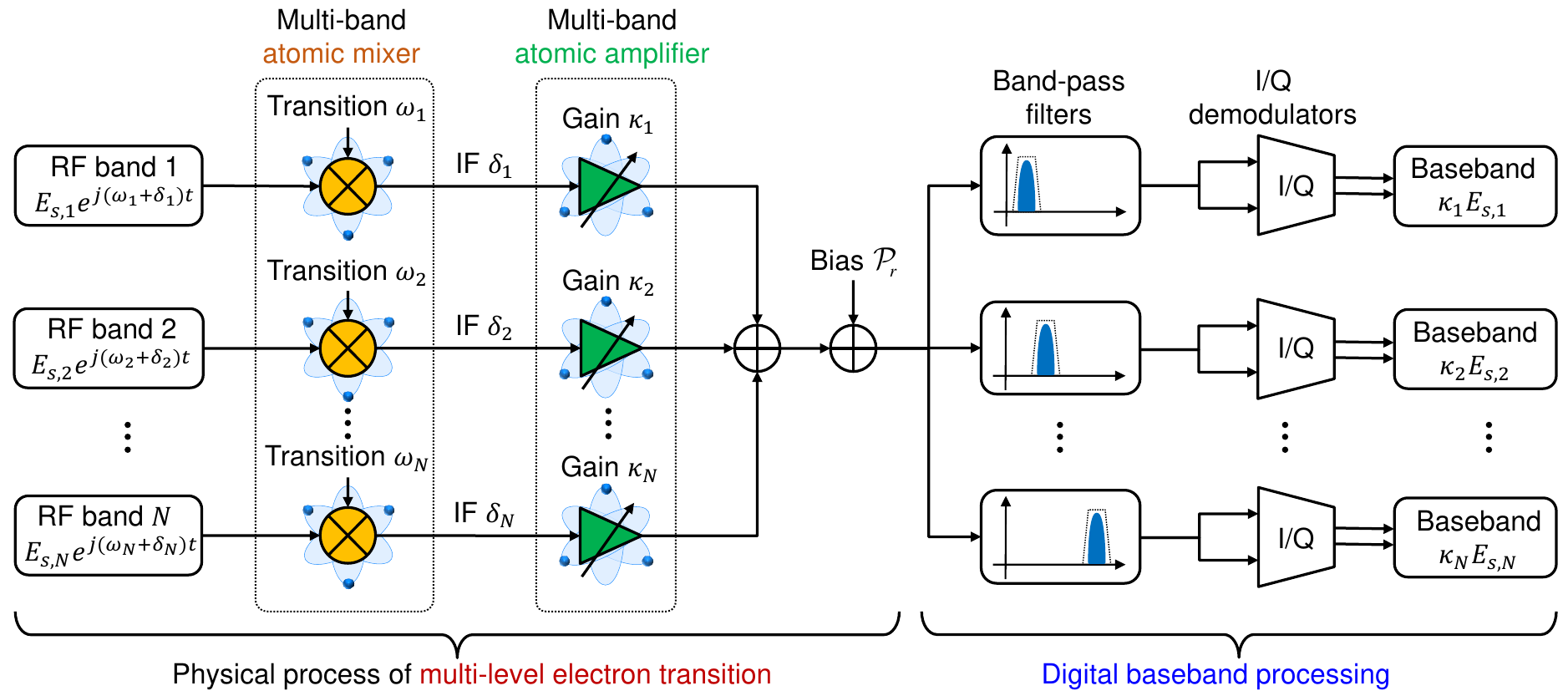}
    \caption{Equivalent signal transmission model of a multi-band RARE.} 
	\vspace*{-1em}
	\label{img:RARE}
\end{figure*}

\begin{remark} \rm
\emph{\textbf{(Physical significance of a multi-band RARE)}}  A multi‑band RARE functions as both a multi‑band atomic mixer and a multi‑band atomic amplifier, as illustrated in Fig.~\ref{img:RARE}. The detection mechanism proceeds as follows:
\begin{itemize}
    \item \textbf{Frequency mixing:} At each band $(\omega_n + \delta_n)$, the data signal couples to the electron transition $\ket{e_0}\rightarrow\ket{e_n}$, inducing coherent down‑conversion via the transition frequency $\omega_{e_0e_n} = \omega_n$. This process produces an IF component at $\delta_n$.
    \item \textbf{Amplification:} Each IF component is then amplified by the band‑specific intrinsic gain $\kappa_n$. This gain factor determines the RARE’s sensitivity to the $n$-th RF band.
    \item \textbf{Signal aggregation:} The amplified IF components from all bands, together with the DC bias, superimpose in the measured photocurrent $I(\boldsymbol{\Omega})$, enabling simultaneous multi-band signal recovery.
\end{itemize}
\end{remark}

The derived linear transfer function is in accordance with the existing single-band~\cite{RydNP_Jing2020,liu_continuous-frequency_2022} and multi-band studies~\cite{MBRARE_Holloway2021, RydMultiband_Du2022, Rydmultiband_Meyer2023, DBRARE_Ding2024, Rydmultiband_Allinson2024, photonics10030328, Frequencyhop_Wen2024}. First, it generalizes the linear transfer function established in the single-band system~\cite{RydNP_Jing2020,liu_continuous-frequency_2022}. When $N = 1$, the band-specific intrinsic gain, $\kappa_n$, naturally reduces to the single-band intrinsic gain derived in~\cite{RydNP_Jing2020,liu_continuous-frequency_2022}.   Second, the transfer function aligns with the multi-band experimental results. For example, the received signal contains multiple IF components, producing multiple spectral peaks as illustrated in Fig.~3 of~\cite{Rydmultiband_Meyer2023}. The gain factor, $\kappa_n$, for a given band depends on the Rabi frequencies (field strengths) across all bands, consistent with the coupled sensitivity observed in~\cite{photonics10030328}.

\subsubsection{Transfer function at the digital baseband}
Digital baseband processing is then applied to convert the measured IF components into the baseband. 
To prevent inter-band interference, it is essential to make the IF components, $\Re{E_{{\rm s},n}e^{j(\delta_n t + \phi_n)}}$, mutually orthogonal. To this end, the 
{separation}
between adjacent IFs, $\frac{\delta_{n + 1} - \delta_n}{2\pi}$, should exceed the average occupied bandwidth, 
$\frac{B_n + B_{n + 1}}{2}$, by a sufficient guard interval. For instance, when $N = 3$, one may choose $\delta_1 = 2\pi\times100\:{\rm kHz}$, $\delta_2 = 2\pi\times200\:{\rm kHz}$, $\delta_3 = 2\pi\times300\:{\rm kHz}$, and set $B_1 = B_2 = B_3 = 50\:{\rm kHz}$. Under this condition, each band’s data signal can be independently recovered using band‑pass filters and I/Q demodulators, as illustrated in the right part of Fig.~\ref{img:RARE}. By further including system noise, the baseband signal originating from the $n$-th RF band thus becomes
\begin{align}
    y_n = \kappa_n E_{{\rm s},n} + z_n,\:\forall n\in\{1,2,\cdots, N\}.
\end{align}
 Here, $z_n$ denotes the noise, whose distribution will be discussed in next sub-section.
The outputs $\{y_n\}$ enable multi-band CommunSense applications, as detailed in Section~\ref{sec:4}.


\subsection{Multi-Band Noise Characterization}
Finally, we characterize the noise affecting the detected signal 
$y_n$
  by extending the single‑band model of~\cite{RydbergNoise_2024}. Following~\cite{RydbergNoise_2024}, the noise 
$z_n$
  comprises two key components: intrinsic noise and extrinsic noise. 

The intrinsic noise arises primarily from photon shot noise (PSN) in the PD~\cite{RydbergNoise_2024}, a consequence of the discrete nature of photon arrivals. {\color{black} The equivalent noise power within bandwidth $B_n$ for band $n$ is proportional to the DC bias of photocurrent, $I_{\rm r}$.} Accordingly, it is given by 
\begin{align}
    \sigma^2_{{\rm I}, n} = q I_{\rm r}B_n.
\end{align}

The extrinsic noise primarily originates from the blackbody radiation at an ambient temperature $T_{\rm a}$ and the vacuum fluctuation~\cite{RydbergNoise_2024}.  
 It contaminates the wireless signal over the air, and is thereby amplified by the intrinsic gain of RARE.  In the multi-band RARE system, we adapt the extrinsic‑noise model of~\cite{RydbergNoise_2024} by replacing the single‑band gain with the band‑specific gain factor $\kappa_n$. {\color{black} The resulting extrinsic noise power at band $n$ is given as (see derivation details in Appendix \ref{app:bbr})
\begin{align}
    \sigma^2_{{\rm E},n} = \kappa_n^2B_n\frac{\hbar\omega_n^3({2n_{\rm th}} + {1})}{\pi \epsilon_0 c^3}\overset{\Delta}{=}\kappa_n^2B_nS_n,
\end{align}
where $\epsilon_0$ is the {\color{black} vacuum permittivity}, $k_B$ the Boltzmann constant, and $n_{\rm th} = {1}/{\left(e^{\hbar \omega_n/k_B T_{\rm a}} - 1\right)}$ the Bose-Einstein distribution.  }


Based on the preceding analysis, the overall noise $z_n$ is modeled as additive
Gaussian noise, i.e., $z_n \sim \mathcal{CN}(0,\sigma_n^2)$, with
\begin{align}
    \sigma_n^2 = \sigma^2_{{\rm I}, n} + \sigma^2_{{\rm E},n} = qI_{\rm r}B_n + \kappa_n^2B_nS_n.  
\end{align}

{\color{black} Notice that technical noise contributions, like the laser intensity noise, PD electrical noise, and the LO phase noise, are not included in our model. While valuable for experimental implementations, these noises are often implementation‑dependent and can be substantially reduced through established techniques (e.g., by locking lasers to ultra‑stable cavities or by employing phase‑stabilization modules). Our focus remains on the fundamental physical noises, the PSN and extrinsic noise, which cannot be mitigated through receiver design alone and therefore set the ultimate sensitivity limit of a Rydberg receiver.
}

\section{Sensitivity Analysis and Maximization for Multi-Band RAREs}\label{sec:3}
In this section, we maximize the sensitivity of a multi-band RARE by optimizing the multi-band atomic amplifier through manipulation of the reference‑signal Rabi frequencies. We first derive an explicit expression for the sensitivity. Then, the concepts of effective Rabi frequency and Rabi attention are introduced to decouple the sensitivity into a tractable form. Finally, we present the optimal effective Rabi frequency that maximizes sensitivity.
\subsection{Sensitivity Analysis}
We use the signal-to-noise ratio (SNR) of the baseband signal, $y_n$, to quantify the sensitivity of a multi-band RARE. For the $n$-th band, the SNR is given by
\begin{align}\color{black}
    {\rm SNR}_n = \frac{\kappa_n^2 |E_{{\rm s},n}|^2}{\kappa_n^2B_nS_n + qI_{\rm r}B_n}.  
\end{align}
It is observed that ${\rm SNR}_n$ increases monotonically w.r.t. the ratio $\frac{\kappa_n^2}{I_{\rm r}}$. In this ratio, 
both the intrinsic gain $\kappa_n = -\frac{\partial I(\boldsymbol{\Omega}_{r})}{\partial \Omega_{{\rm r},n}}\frac{\mu_n}{\hbar}$ and DC bias $I_{\rm r} = I(\boldsymbol{\Omega}_r)$ are affected by the Rabi frequency vector $\boldsymbol{\Omega}_{r}$ of reference signals. Hence, the sensitivity of RARE can be maximized by appropriately configuring the values of  $\boldsymbol{\Omega}_{r}$ through the transmission power control at the reference sources.
However, the multi-band gain factors, $\{\kappa_1,\kappa_2,\cdots,\kappa_N\}$,  are mutually coupled due to their intricate dependence on all Rabi frequencies, as shown in \eqref{eq:y1} and \eqref{eq:gradient}. This coupling makes the maximization of ${\rm SNR}_n$ nonconvex and challenging. 
To address this issue, we propose a method to decouple the gain factor $\kappa_n$ from its joint dependence on $\boldsymbol{\Omega}_r$, thereby enabling tractable optimization.

\subsubsection{Decoupling of the gain factor $\kappa_n$} The concepts of \emph{effective Rabi frequency} and \emph{Rabi attentions} are introduced to decouple the gain factor $\kappa_n$.  

\begin{definition}[Effective Rabi frequency and Rabi attentions]\rm 
The effective Rabi frequency is defined as the root‑sum‑square of the Rabi frequencies of the reference signals: 
    \begin{align}
 \Omega_{\rm eff} \overset{\Delta}{=} \left({\sum_{n = 1}^N \Omega_{{\rm r},n}^2}\right)^{1/2}.
\end{align}
The Rabi attention for the $n$-th band is defined as the fraction of the squared Rabi frequency contributed by that band:
\begin{align}
    \alpha_n\overset{\Delta}{=}\frac{\Omega_{{\rm r},n}^2}{\sum_{m = 1}^N\Omega_{r,m}^2} = \frac{\Omega_{{\rm r},n}^2}{\Omega_{\rm eff}^2}, \forall n,
\end{align}
which satisfies $\sum_{n=1}^N \alpha_n = 1$.
\end{definition}

{\color{black}
\begin{remark}[Physical significance of effective Rabi frequency and Rabi attention]\rm 
The effective Rabi frequency quantifies the cumulative strength of all reference signals driving electron transitions between Rydberg states. 
The Rabi attention $\alpha_n$ describes how this overall strength is distributed among frequency bands and, correspondingly, how the electron population is distributed among the Rydberg levels $\ket{e_n}$. A larger $\alpha_n$ indicates a stronger reference signal at the $n$-th band, which promotes a higher fraction of electrons into the associated state $\ket{e_n}$.
This is reflected in the steady-state populations $\rho_{n+3,n+3}$ (see Appendix~A), which are proportional to their respective Rabi attentions $\alpha_n$:
\begin{align}
    \rho_{n + 3,n + 3} = C \alpha_n, 
\end{align}
where $C = \frac{\Omega_{\rm p}^2(\Omega_{\rm p}^2 + \Omega_{\rm c}^2)}{(\gamma_2^2 + 2\Omega_{\rm p}^2)\Omega_{\rm eff}^2 + 2\Omega_{\rm p}^2(\Omega_{\rm p}^2 + \Omega_{\rm c}^2)}$ is normalization constant. Thus, increasing $\alpha_n$ enhances the occupation of level $\ket{e_n}$ relative to other Rydberg levels, demonstrating how Rabi attention shapes the energy‑level occupation.
\end{remark}
}

Clearly, the optimization of Rabi frequencies $\{\Omega_{{\rm r},1}, \Omega_{{\rm r}, 2}, \cdots, \Omega_{{\rm r}, N}\}$ is equivalent to optimizing the effective Rabi frequency together with the Rabi attentions. This is attributed to the one-to-one mapping 
\begin{align}
    \Omega_{{\rm r},n} = \Omega_{\rm eff}\sqrt{\alpha_n},\:\forall n. 
\end{align}
Utilizing this equivalence, the DC bias $I_{\rm r}$ depends solely on the effective Rabi frequency:
\begin{align}\label{eq:laserpower2}
    I_{\rm r} = I(\boldsymbol{\Omega}_r) = P_{\rm in}\exp\left(-\frac{\chi_0\Omega_{\rm eff}^2}{\Omega_{\rm eff}^2 + \Gamma^2}\right).
\end{align} 
Moreover, the intrinsic gain, $\kappa_n$, for band $n$ can be recast as 
\begin{align}\label{eq:decoupling_kappa}
    \kappa_n &= I_{\rm r}\frac{2\chi_0\Gamma^2\Omega_{{\rm r},n}}{(\sum_{m=1}^N\Omega_{r,m}^2 + \Gamma^2)^2}\frac{\mu_n}{\hbar} =I_{\rm r}\frac{2\chi_0\Gamma^2\Omega_{\rm eff}\sqrt{\alpha_n}}{(\Omega_{\rm eff}^2 + \Gamma^2)^2}\frac{\mu_n}{\hbar}, \notag \\
    & =\frac{2\chi_0\Gamma^2}{\hbar}I_{\rm r}\frac{{\Omega_{\rm eff}}}{(\Omega_{\rm eff}^2+\Gamma^2)^2}\sqrt{\alpha_n}\mu_n \overset{\Delta}{=} \varrho_0\sqrt{\alpha_n}\mu_n,
\end{align}
where  
\begin{equation}\label{eq:Global_gain}
    \varrho_0\overset{\Delta}{=}\frac{2\chi_0\Gamma^2}{\hbar}I_{\rm r}\frac{\Omega_{\rm eff}}{(\Omega_{\rm eff}^2+\Gamma^2)^2}.
\end{equation}
Equation \eqref{eq:decoupling_kappa} decouples the intrinsic gain factor  into three components: the transition dipole moment $\mu_n$, the parameter $\varrho_0$ termed as the \emph{global gain}, and {\color{black} the square root of Rabi attention $\sqrt{\alpha_n}$}. This decomposition enables the decoupled representation of the multi‑band atomic amplifier illustrated in Fig.~\ref{img:decoupling}, which we interpret in Remark~3.

\begin{figure}
    \centering
    \includegraphics[width=3.5in]{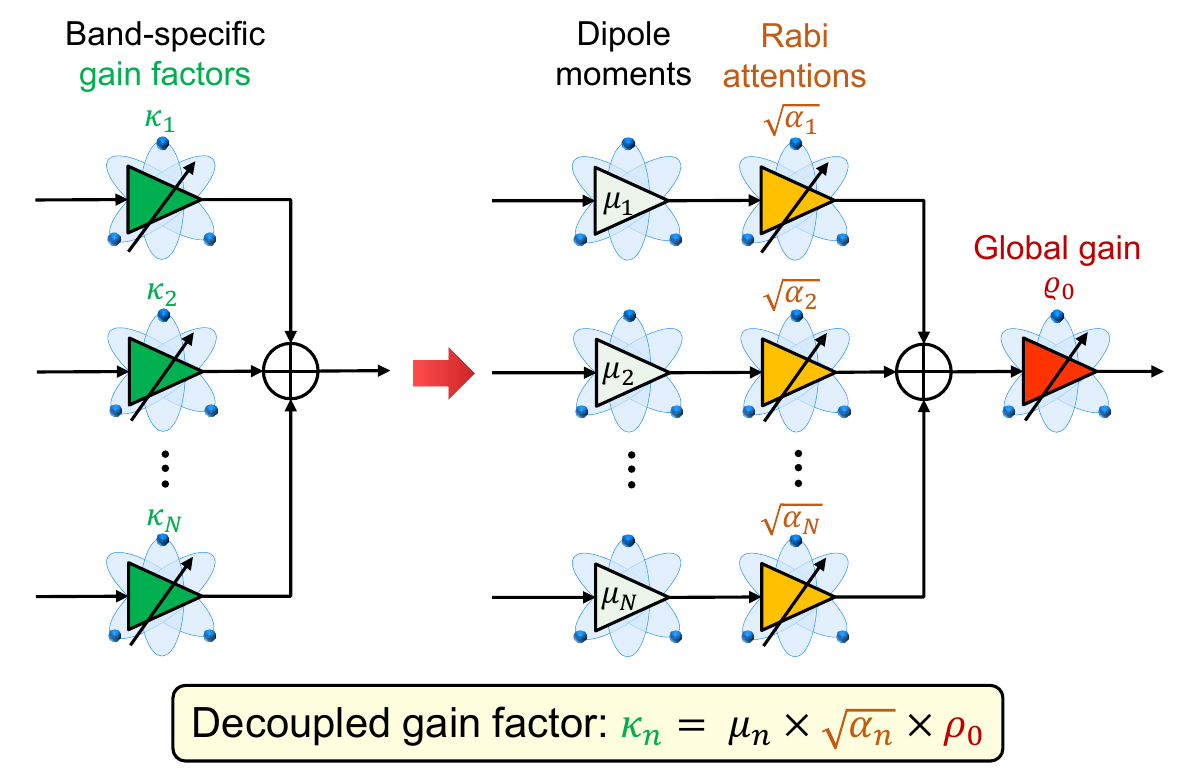}
    \vspace*{-1em}
    \caption{\color{black} Decoupling of the multi-band atomic amplifier.} 
	\vspace*{-1em}
	\label{img:decoupling}
\end{figure}

\begin{remark}[Decoupled band-specific gain factors]\rm~
\begin{itemize}
    \item \textbf{Global gain:} The global gain $\varrho_0$ depends solely on the effective Rabi frequency $\Omega_{\rm eff}$, which remains consistent across all frequency bands. This parameter determines the overall amplification gain, or equivalently the overall sensitivity, of the RARE to all frequency bands. 
    Because the effective Rabi frequency $\Omega_{\mathrm{eff}}$ and the Rabi attentions $\alpha_n$ are independent, the global gain can be optimized separately from the band‑specific attention weights.
    \item \textbf{Square root of Rabi attention:} The Rabi attentions, $\{\alpha_n\}$, dictate how the RARE allocates its overall sensitivity to different frequency bands.  This allocation constitutes a distinctive
    \emph{attention mechanism} enabled by multi-level electron transitions. 
    Increasing $\alpha_n$ raises the electron population in the corresponding Rydberg state $\ket{e_n}$ relative to other states, thereby enhancing sensitivity to the $n$‑th band while reducing sensitivity to other bands.
    This explains the experimentally observed sensitivity tradeoff of multi-band RAREs~\cite{photonics10030328}.
    \item \textbf{Transition dipole moment:}  
    {The transition dipole moments, $\{\mu_n\}$, are constants determined by the Rydberg atoms' inherent response to each band. These parameters remain fixed, reflecting fundamental atomic properties.}
\end{itemize} 
\end{remark}



As a result, the decoupling of band-specific gain factors allows us to \emph{independently} optimize the global gain, $\varrho_0$, and the Rabi attentions, $\{\alpha_n\}$. This is instrumental in solving the sensitivity maximization problem.

\subsubsection{SNR analysis}
Following the decoupled $\kappa_n$ in \eqref{eq:decoupling_kappa}, the SNR at band $n$ is reformulated as 
\begin{align}\label{eq:SNR2}
    {\rm SNR}_n = \frac{\alpha_n\mu_n^2 |E_{{\rm s},n}|^2}{\alpha_n \mu_n^2 B_n S_n + \frac{I_{\rm r}}{\varrho_0^2}qB_n}. 
\end{align}
Several insights can be drawn from \eqref{eq:SNR2}. {\color{black} Firstly, the SNR is monotonically increasing w.r.t. the global gain normalized by the DC bias, $\frac{\varrho_0^2}{I_{\rm r}}$.} 
The normalization arises from the PSN, that is linearly related to $I_{\rm r}$. Similar to the discussion in Remark~3, the normalized global gain, $\frac{\varrho_0^2}{I_{\rm r}}$, serves as a global parameter dependent only on the effective Rabi frequency. Increasing its value allows for a uniform enhancement of the SNRs across all frequency bands. Therefore, the maximization of all ${\rm SNR}_n,\forall n$ is achieved only if $\frac{\varrho_0^2}{I_{\rm r}}$ is maximized. 
Secondly, the Rabi attentions, $\alpha_n, \forall n$, dictate the RARE's emphasis on each frequency band. By increasing $\alpha_n$, the RARE can amplify the signal, $E_{{\rm s},n}$, with a larger gain. However, this also results in an increase in the extrinsic noise, {\color{black} $\sigma_{{\rm E},n}^2$}, within the same band. 


\subsection{Sensitivity Maximization} 
The optimal effective Rabi frequency $\Omega_{\rm eff}=\sqrt{\sum_{n=1}^N\Omega_{{\rm r},n}^2}$ is derived in the sequel to maximize the normalized global gain  $\frac{\varrho_0^2}{I_{\rm r}}$, {thereby achieving} the maximal {sensitivity across} all frequency bands. 
Notably, the optimization of the Rabi attentions {$\{\alpha_n\}$} is 
contingent on application-specific performance metrics for CommunSense, which is detailed in Section~\ref{sec:4}.

Specifically, the optimal effective Rabi frequency that maximizes the sensitivity obeys
\begin{align}\label{eq:P0}
    \max_{\Omega_{\rm eff} > 0}\:\frac{\varrho_0^2}{I_{\rm r}},
\end{align}
where 
\begin{align}\label{eq:objective}
    \frac{\varrho_0^2}{I_{\rm r}} = \frac{4\chi_0^2\Gamma^4 P_{\rm in}}{\hbar^2}\exp\left(-\frac{\chi_0\Omega_{\rm eff}^2}{\Omega_{\rm eff}^2 + \Gamma^2}\right)\frac{\Omega_{\rm eff}^2}{(\Omega_{\rm eff}^2 + \Gamma^2)^4}.
\end{align}
Theorem~2 provides the closed-form global optimum for \eqref{eq:P0}, which governs the fundamental sensitivity limit of a multi-band RARE. 
\begin{theorem}\rm
    The optimal $\Omega_{\rm eff}$ to the problem \eqref{eq:P0} is given as 
    \begin{align}\label{eq:A}
        \Omega_{\rm eff}^\star =  \sqrt{\frac{\chi_0+4 - \sqrt{\chi_0^2 + 4\chi_0 + 16}}{\chi_0 -4 + \sqrt{\chi_0^2 + 4\chi_0 + 16}}}\Gamma. 
    \end{align}
\end{theorem}
\begin{proof}
    (See Appendix C). 
\end{proof}
Utilizing Theorem 2, the normalized global gain $(\frac{\varrho_0}{I_{\rm r}})^{\star}$  that maximizes sensitivity is acquired. By substituting $\Omega_{\rm eff}^\star$ into \eqref{eq:laserpower2} and \eqref{eq:Global_gain}, we can also obtain the corresponding DC bias and global gain, denoted as $I_{\rm r}^{\star}$ and $\varrho_0^{\star}$ respectively. 

{\color{black} 
\begin{figure}
    \centering
    \includegraphics[width=3.3in]{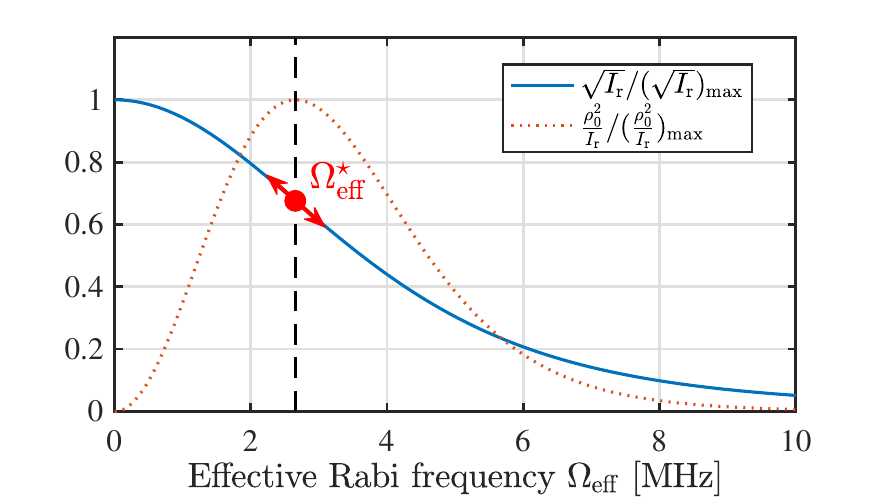}
    \caption{The influence of effective Rabi frequency on the $\sqrt{I_{\rm r}}$ curve and the normalized global gain ${\rho_0^2}/{I_{\rm r}}$, both of which are normalized by their maximal values $(\sqrt{I_{\rm r}})_{\rm max}$ and $({\rho_0^2}/{I_{\rm r}})_{\rm max}$.
    } 
	\vspace*{-1em}
	\label{img:Rabi_eff}
\end{figure}

To gain deeper insights into the optimal effective Rabi frequency, we derive an equivalent expression of $\frac{\varrho_0^2}{I_{\rm r}}$: 
\begin{align}\label{eq:insight_effective}
    \frac{\varrho_0^2}{I_{\rm r}} = \frac{4P_{\rm in}}{\hbar^2} \left(\frac{{\rm d}}{{\rm d} \Omega_{\rm eff}}e^{-\frac{\chi_0\Omega_{\rm eff}^2}{2(\Omega_{\rm eff}^2 + \Gamma^2)}}\right)^2 \propto \left(\frac{{\rm d}\sqrt{I_{\rm r}}}{{\rm d} \Omega_{\rm eff}}\right)^2.
\end{align}
Equation \eqref{eq:insight_effective} indicates that the optimal effective Rabi frequency captures the \emph{steepest slope} of the square root of DC bias, $\sqrt{I_{\rm r}}$, as presented in Fig.~\ref{img:Rabi_eff}. This is because the data signals are small perturbations relative to the reference counterparts, causing the probe laser to fluctuate around the DC bias. The fluctuation amplitude, which reflects a RARE's sensitivity, is intrinsically linked to the gradient of the DC bias $I_{\rm r}$. Moreover, as the power of PSN is also proportional to the DC bias, the gradient of $I_{\rm r}$ is calibrated as the gradient of $\sqrt{I_{\rm r}}$ when measuring SNR. Consequently, the optimal sensitivity is achieved where the rate of change in $\sqrt{I_{\rm r}}$ w.r.t $\Omega_{\rm eff}$ is maximized, not simply where the change in $I_{\rm r}$ is largest. This approach connects naturally with established single-band optimization strategies.
In the prior single-band study~\cite{RydNP_Jing2020}, the Rabi frequency was optimized by seeking the steepest point on the EIT curve, which directly relates to the curve of DC bias. The method inherently avoids EIT saturation by operating away from the fully transparent regime. Our discussion advances~\cite{RydNP_Jing2020} by generalizing Rabi frequency to the effective counterpart and balancing the impacts of EIT response (DC bias) and PSN. 
}

Furthermore, Theorem 2 provides a practical guideline for implementing multi-band reference sources: achieving maximum sensitivity requires the weighted sum of reference signal powers to be constant. This is evident from the definition of effective Rabi frequency:
\begin{align}\color{black}
   (\Omega_{\rm eff}^\star)^2 =  \sum_{n = 1}^N\frac{\mu_n^2}{\hbar^2}|E_{{\rm r},n}|^2  = \text{constant}.
\end{align}
This condition arises because the cumulative contributions from all frequency bands must precisely tune the effective Rabi frequency to the steepest point on the $\sqrt{I_{\rm r}}$ curve.
Allowing any individual frequency band to become too strong would drive the system away from this optimum.
On the other hand, although the effective Rabi frequency is constrained, the distribution of reference signal powers among frequency bands remains adjustable.  By manipulating the ratio ${\Omega_{{\rm r},n}^2}/{\Omega_{\rm eff}^2}$, we can tailor the Rabi attentions ${\alpha_n}$, creating additional degrees of freedom for further optimizing CommunSense performance.


\section{Applications of Multi-Band RAREs}\label{sec:4}
In this section, we apply multi-band RAREs to {CommunSense} scenarios. 
The optimal effective Rabi frequency $\Omega_{\rm eff}^\star$ derived in Section \ref{sec:3}-B is adopted, while the optimal Rabi attentions $\{\alpha_n\}$ for each CommunSense scenario are systematically designed. Our analysis establishes the performance limits for concurrent detection of the CommunSense information, $\{s_n\}$ in \eqref{eq:sn} and $\{d_n\}$ in \eqref{eq:dn}, by RAREs.

\subsection{Multi-Band Communications with RAREs}
Consider the multi-band communication system. A single RARE is employed to capture information from $N$ heterogeneous communication devices.
Leveraging the derived transfer function of multi-band RARE with the maximal sensitivity, the received communication signal at band $n$ is obtained as
\begin{align}\label{eq:}
    y_n = \varrho_0^\star \sqrt{\alpha_n}\mu_n {h}_ns_n +z_n.
\end{align}
Following this signal model, the multi-band symbols $\{s_n\}$ can be recovered through the maximum-likelihood (ML) {detector}:
\begin{align}
    \hat{s}_n = \arg\min_{s_n \in \mathbb{S}_n}\left\|{y_n} - \varrho_0^\star \sqrt{\alpha_n}\mu_n {h}_n s_n \right\|,\forall n,
\end{align}

To evaluate this multi-band communication system, we use the SE as the performance metric. The established signal model allows us to explicitly express the SE as 
\begin{align}
    \mathcal{C} &= {\sum_{n = 1}^NB_n\log_2(1 + {\rm SNR}_n)}/{\sum_{n = 1}^N B_n} \notag\\&= \sum_{n = 1}^N\gamma_n\log_2\left(1 + \frac{\alpha_n\mu_n^2{\varrho_0^\star}^2|{h}_n|^2}{\alpha_n\mu_n^2{\varrho_0^\star}^2 B_n S_n + qI_{\rm r}^\star B_n}\right),
\end{align}
where $\gamma_n\overset{\Delta}{=}\frac{B_n}{\sum_{m=1}^NB_m}$ denotes the proportion of bandwidth. 
For ease of discussion, the transmission power $P_{{\rm s},n}$ and bandwidth $B_n$ are assumed to be fixed. We aim at maximizing the SE by optimizing Rabi attentions, $\{\alpha_n\}$, {under the constraints $\sum_{n=1}^N\alpha_n=1$ and $\alpha_n \ge 0, \forall n$}, to approach the performance limit of multi-band communications. 
{Using the definitions $\beta_n \overset{\Delta}{=}\frac{|{h}_n|^2}{B_n S_n}$ and $\epsilon_n \overset{\Delta}{=}\frac{qI_{\rm r}^\star}{\mu_n^2{\varrho_0^\star}^2S_n}$ can simplify the problem to:}
\begin{align}
    \max_{{\{\alpha_n \ge 0\}_{n=1}^N}} \:\:&\sum_{n=1}^N\gamma_n\log_2\left(1 + \frac{\beta_n\alpha_n}{\alpha_n + \epsilon_n}\right), \label{eq:ASE} \\
    {\rm s.t.}\:\:&\sum_{n=1}^N\alpha_n = 1, \tag{\ref{eq:ASE}a} \label{eq:ASEa}
\end{align}
Notice that the formulated attention allocation problem is analogous to classical multiple-input-multiple-output (MIMO) power allocation. The key distinctions lie in the conversion from total power constraint to the total attention constraint $\sum \alpha_n = 1$, as well as the influence of external-noise power, which is scaled by the attention weight $\alpha_n$. 

To establish that a global optimum exists, we prove the convexity of problem \eqref{eq:ASE}. First, the equality $\sum_{n=1}^N\alpha_n=1$ and the non-negativity conditions $\alpha_n\geq0$ are linear and thus convex.
Then, the Hessian matrix is negative definite because the second partial derivatives of the objective function are 
\begin{align}\label{eq:H1}
\color{black}\frac{\partial{\mathcal{C}^2}}{\partial \alpha_n^2} = -\frac{\gamma_n\beta_n\epsilon_n(2(1\!+\!\beta_n)\alpha_n\! +\! (2\!+\!\beta_n)\epsilon_n)}{(\alpha_n + \epsilon_n)^2((1+\beta_n)\alpha_n+\epsilon_n)^2\ln 2\:}<0,
\end{align}
and $\frac{\partial{\mathcal{C}^2}}{\partial \alpha_n\partial\alpha_m} = 0, \:\:\forall n\neq m$.  Therefore, problem \eqref{eq:ASE} is a convex optimization problem, and its global optimum can be obtained by solving the Karush-Kuhn-Tucker (KKT) conditions, leading to the following theorem.
\begin{theorem}\rm
    The optimal Rabi attentions, $\alpha_n,\:\forall n\in\{1,2,\cdots, N\}$, that solve problem \eqref{eq:ASE} satisfy
    \begin{align}
        \alpha_n^\star \!=\! \left(\frac{-(2\!+\!\beta_n)\epsilon_n \!+\! \sqrt{\beta_n^2\epsilon_n^2 \!+\! 4\nu^\star\gamma_n\beta_n(1\!+\!\beta_n)\epsilon_n}}{2(1\!+\!\beta_n)} \right)^+\!\!\!,
\end{align}
where $(x)^+ \overset{\Delta}{=}\max(0, x)$, and the non-negative Lagrange multiplier $\nu^\star$ is properly selected to ensure $\sum_{n = 1}^N\alpha_n^\star = 1$.
\end{theorem}
\begin{proof}
    (See Appendix D).
\end{proof}
Due to its mathematical structure, Theorem 3 can be interpreted as a \emph{receiver-side water-filling principle}, analogous to classical power allocation in MIMO systems but adapted to the quantum domain of attention distribution.
Combining Theorem 2 and Theorem 3 yields the optimal design of Rabi frequencies {of the reference signals:} $ \Omega_{{\rm r},n}^\star = \Omega_{\rm eff}^\star\sqrt{\alpha_n^\star}$, where $\Omega_{\rm eff}^{\star}$ (Theorem 2) ensures maximal sensitivity, and $\alpha_n^{\star}$ (Theorem 3) optimizes SE. This joint optimization enables the multi-band RARE to achieve the maximal capacity in concurrent multi-band communications. 

It needs to be emphasized that the metric of SE is just one possible criterion for multi-band RAREs. Many other metrics could also be considered, such as energy efficiency, user fairness, and secure communication rate. These metrics will result in different attention distribution, while the optimal effective Rabi frequency can be kept as \eqref{eq:A} because its design is independent of the specific metrics.

\subsection{Multi-Band Sensing with RAREs}
{Multi-band RAREs can enable multi-granularity wireless sensing by exploiting the frequency-dependent phase sensitivity of multi-band signals.} In this scenario, 
our objective is to recover the displacements, $\{d_n\}$, of all sensing targets from the received signal, $\{y_n\}$, formulated as 
\begin{align}\label{eq:yn_sense}
    y_n = \varrho_0^\star \sqrt{\alpha_n}\mu_n {h}_ne^{-j\frac{\omega_n}{c}d_n} +z_n.
\end{align}
To this end, we use the ML detector to estimate $d_n$ as 
\begin{align}
    \hat{d}_n = -\frac{c}{\omega_n}\angle\left(\frac{y_n}{\varrho_0^\star\sqrt{\alpha}_n\mu_n{h}_n}\right), \forall n.
\end{align}

{\color{black} The sensing accuracy of $\hat{d}_n$ can be assessed using NCRLB, defined as the Cramér-Rao lower bound of $d_n$ normalized by the variance of $d_n$.} The NCLRB offers a theoretical lower bound of the normalized mean square error when estimating $d_n$. To be specific, we assume $d_n$ a uniform distribution, $d_n \sim \mathcal{U}[-\frac{\pi c}{\omega_n}, \frac{\pi c}{\omega_n}]$, yielding $\mathbb{E}(|d_n|^2) = \frac{\pi^2 c^2}{3\omega_n^2}$. Following the derived signal model in \eqref{eq:yn_sense}, the NCRLB of $d_n$ is formulated as 
\begin{align}
   {\rm NCRLB}(d_n) &= \frac{{\rm CRLB}(d_n)}{\mathbb{E}(|d_n|^2)}
    = \frac{1}{\mathbb{E}(|d_n|^2)}
   \frac{c^2}{\omega_n^2}\frac{1}{2{\rm SNR}_n},\\
   &= \frac{3B_nS_n}{2\pi^2|{h}_n|^2} + \frac{ 3qI_{\rm r}B_n}{2\pi^2\alpha_n\mu_n^2{\varrho_0^\star}^2 |{h}_n|^2}. \notag
\end{align}
The sum of all NCRLBs is adopted as a performance metric to evaluate the overall sensing error, which is minimized by optimizing the Rabi attentions. 
The notation $\xi_n \overset{\Delta}{=} \frac{ 3qI_{\rm r}B_n}{2\pi^2\mu_n^2{\varrho_0^\star}^2 |{h}_n|^2}$ is introduced to simplify the expression of the formulated problem as 
\begin{align}
    \min_{{\{\alpha_n \ge 0\}_{n=1}^N}} \:\:\sum_{n=1}^N\frac{\xi_n}{\alpha_n},  
    \:\:{\rm s.t.}\sum_{n=1}^N\alpha_n = 1,  \label{eq:CRLB} 
\end{align}
The term $\frac{3B_nS_n}{2\pi^2 |{h}_n|^2}$ in ${\rm NCRLB}(d_n)$ is omitted as it is irrelevant to $\alpha_n$. 
This optimization problem is shown to be convex, and its global optimal solution can be readily derived 
by solving KKT conditions
as 
\begin{align}\label{eq:alpha2}
    \alpha_n^\star = \frac{\sqrt{\xi_n}}{\sum_{m = 1}^N\sqrt{\xi_m}}. 
\end{align}
As a result, by configuring the Rabi frequencies as $\Omega_{{\rm r},n}^\star = \Omega_{\rm eff}^\star\sqrt{ {\alpha}^\star_n}$ using the optimal effective Rabi frequency in Theorem 2 and the optimal Rabi attentions in \eqref{eq:alpha2}, the maximal sensing accuracy averaged over multiple frequency bands is achieved.

\section{Validation Results}\label{sec:6}

 In this section, we present validation results.  In sub-section \ref{sec:6A}, we conduct experiments on a multi-band RARE hardware platform to validate the discovered Rabi attention mechanism. Then, in sub-sections \ref{sec:6B} and \ref{sec:6C}, simulations are 
 performed to numerically validate the derived signal model and the multi-band CommunSense performance of RAREs. 

{\color{black}

\subsection{Experimental Validation for Rabi Attention Mechanism}\label{sec:6A}
This subsection presents a proof-of-principle experiment designed to validate the key theoretical findings of our study.
While the multi-band mixing capability of RAREs has been extensively documented in prior works~\cite{MBRARE_Holloway2021, RydMultiband_Du2022, Rydmultiband_Meyer2023, DBRARE_Ding2024, Rydmultiband_Allinson2024, photonics10030328, Frequencyhop_Wen2024}, we focus here on investigating a novel aspect of multi-band RARE systems: the attention mechanism.

Our objective is to experimentally verify the linear relationship between the squared intrinsic gain, $\kappa_n^2$, and the Rabi attention, $\alpha_n = \Omega_{{\rm r}, n}^2/\Omega_{\rm eff}^2$. This relationship can be equivalently demonstrated by showing that the DC bias $I_{\rm r}$ depends exclusively on the effective Rabi frequency $\Omega_{\rm eff} = \sqrt{\sum_{n = 1}^N\Omega_{{\rm r}, n}^2}$, as evidenced by the following derivation: 
\begin{align}\notag
    \kappa_n^2 \propto \left(\frac{\partial I_{\rm r}}{\partial \Omega_{{\rm r}, n}}\right)^2 \overset{(a)}{=} \left(\frac{\partial I_{\rm r}}{\partial \Omega_{\rm eff}^2} 
    \frac{\partial \Omega_{\rm eff}^2}{\partial \Omega_{{\rm r}, n}}
    \right)^2 = \left(2\frac{\partial I_{\rm r}}{\partial \Omega_{\rm eff}^2} \Omega_{\rm eff}
    \right)^2\alpha_n.
\end{align}
Here, equation (a) holds if and only if  $I_{\rm r}$ relates to $\Omega_{{\rm r}, n}$ through $\Omega_{\rm eff}^2$. Furthermore, noting that the square of Rabi frequency in each band $n$ is linearly proportional to the transmit power $P_{{\rm r}, n}$ of the reference source:
$\Omega_{{\rm r}, n}^2 = \frac{\mu_n^2}{\hbar^2} |E_{{\rm r}, n}|^2 \propto P_{{\rm r}, n}$, we can reformulate $\Omega_{\rm eff}$ as $\Omega_{\rm eff} = \sqrt{\sum_{n = 1}^Nc_nP_{{\rm r}, n}}$, where $c_n$ is a constant weight. The proof-of-principle experiment thus reduces to verifying that the DC bias $I_{\rm r}$ is a function of the weighted sum of reference source powers
\begin{align}
    I_{\rm r}  = I_{\rm r}\left(\sum_{n =1}^Nc_nP_{{\rm r}, n}
    \right).
\end{align}
Equivalently, the contours of $I_{\rm r}$ in the variable space $\{P_{{\rm r}, 1}, P_{{\rm r}, 2}, \cdots, P_{{\rm r}, N}\}$ should form straight lines described by \begin{align}
    \sum_{n =1}^Nc_nP_{{\rm r}, n} = W,
\end{align}
where $W$ controls the contour offset.

\subsubsection{Hardware Platform}
\begin{figure}
    \centering
    \includegraphics[width=3in]{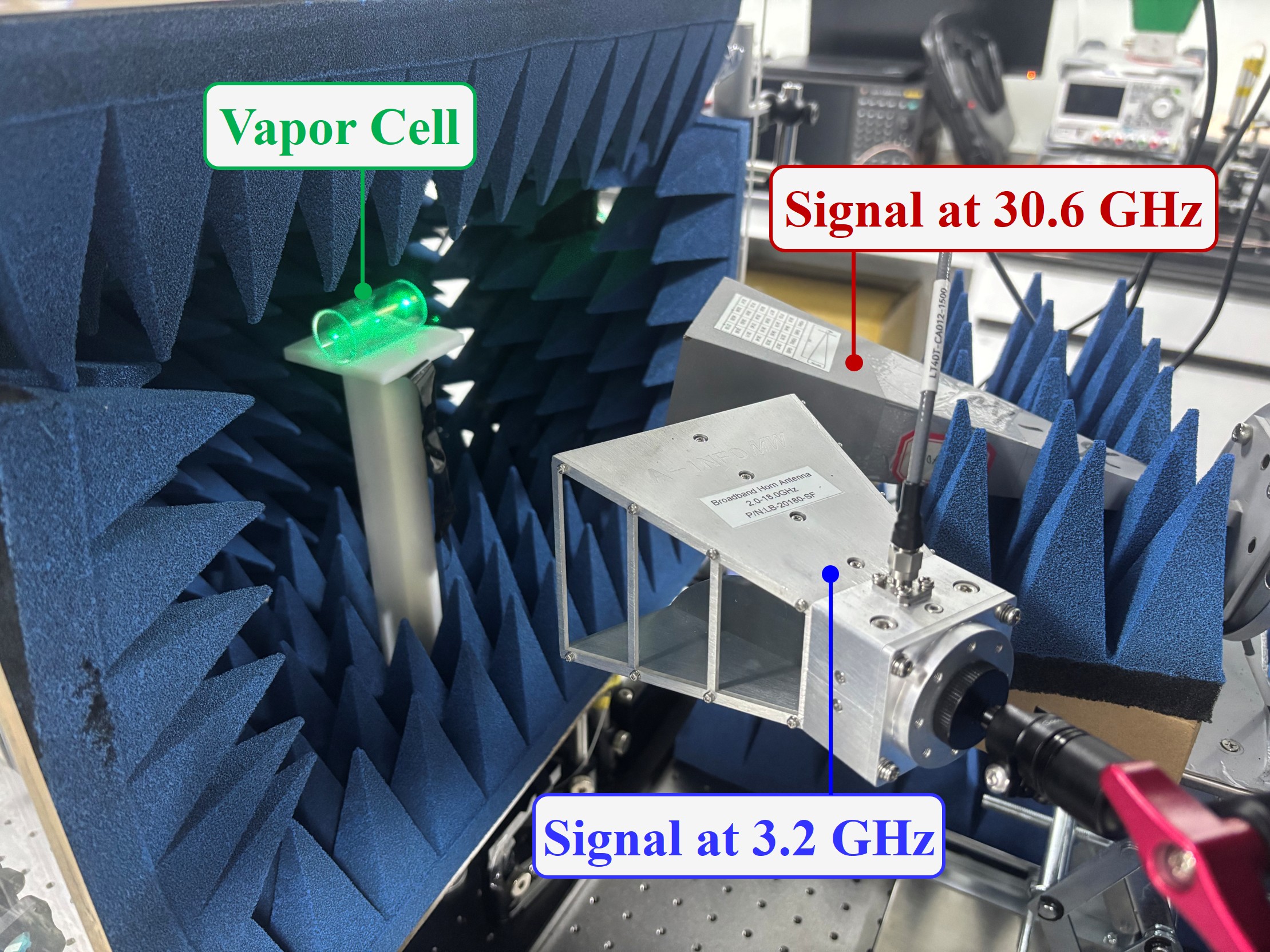}
    \caption{Experiment platform.} 
	\vspace*{-1em}
	\label{img:exp_setup}
\end{figure}

The dual-band experimental setup is illustrated in Fig.~\ref{img:exp_setup}. A cylindrical cesium-133 vapor cell of length 5~cm and diameter 2.5~cm is used. Experiments are conducted in a microwave anechoic chamber to mitigate environmental electromagnetic interference. A probe laser at $852\:{\rm nm}$ excites atoms from the ground state $\ket{g_1} = 6S_{1/2}$ to the lowly excited state $\ket{g_2} = 6P_{3/2}$. A 509~nm coupling laser further drives the transition to the initial Rydberg state $\ket{e_0} = 60D_{5/2}$. The probe laser operates at a power of $100 \mu {\rm W}$ with a beam waist of 1.0~mm, while the coupling laser delivers $40{\rm mW}$ with a beam waist of 1.3~mm. Two signal generators (Ceyear 1466H-V and 1466H-G) serve as the sub-6G and mmWave reference sources. They respectively operate at $\omega_{e_0e_1} = 2\pi \times 3.212\:{\rm GHz}$ and $\omega_{e_0e_2} = 2\pi\times 30.618\:{\rm GHz}$, which trigger transitions to Rydberg states $\ket{e_1} = 61P_{3/2}$ and $\ket{e_2} = 62P_{3/2}$. The RF signals are emitted via two horn antennas (LB-20180-SF for sub-6G and LB-34-25-A2 for mmWave), placed 20~cm from the vapor cell, with their polarization aligned parallel to that of the coupling laser. Variations in the probe laser transmission induced by the dual-band signals are converted into a DC bias $I_{\rm r}$ via a photodetector (PDA36A2, Thorlabs). This signal is recorded as an output voltage $V_{\rm r} = I_{\rm r}R_{\rm L}$ by a digital oscilloscope (MSO56B, Tektronix) at a sampling rate of $100{\rm MS/s}$, where $R_{\rm L} = 2{\rm k\Omega}$ is the load resistance.  
During experiments, the transmit power $P_{{\rm r}, 1}$ of the sub-6G source is varied from $10\:\mu{\rm W}$ to $50\:\mu{\rm W}$, while that at the mmWave band, $P_{{\rm r}, 2}$, is increased from $50\:\mu{\rm W}$ to $150\:\mu{\rm W}$. The resulting DC voltages $V_{\rm r}$ are recorded to characterize their dependence on the transmit powers $\{P_{{\rm r}, 1}, P_{{\rm r}, 2}\}$.

\subsubsection{Experiment results}
\begin{figure}[t!]
    \centering
    \subfigure[3D view]
    {\includegraphics[width=0.49\columnwidth]{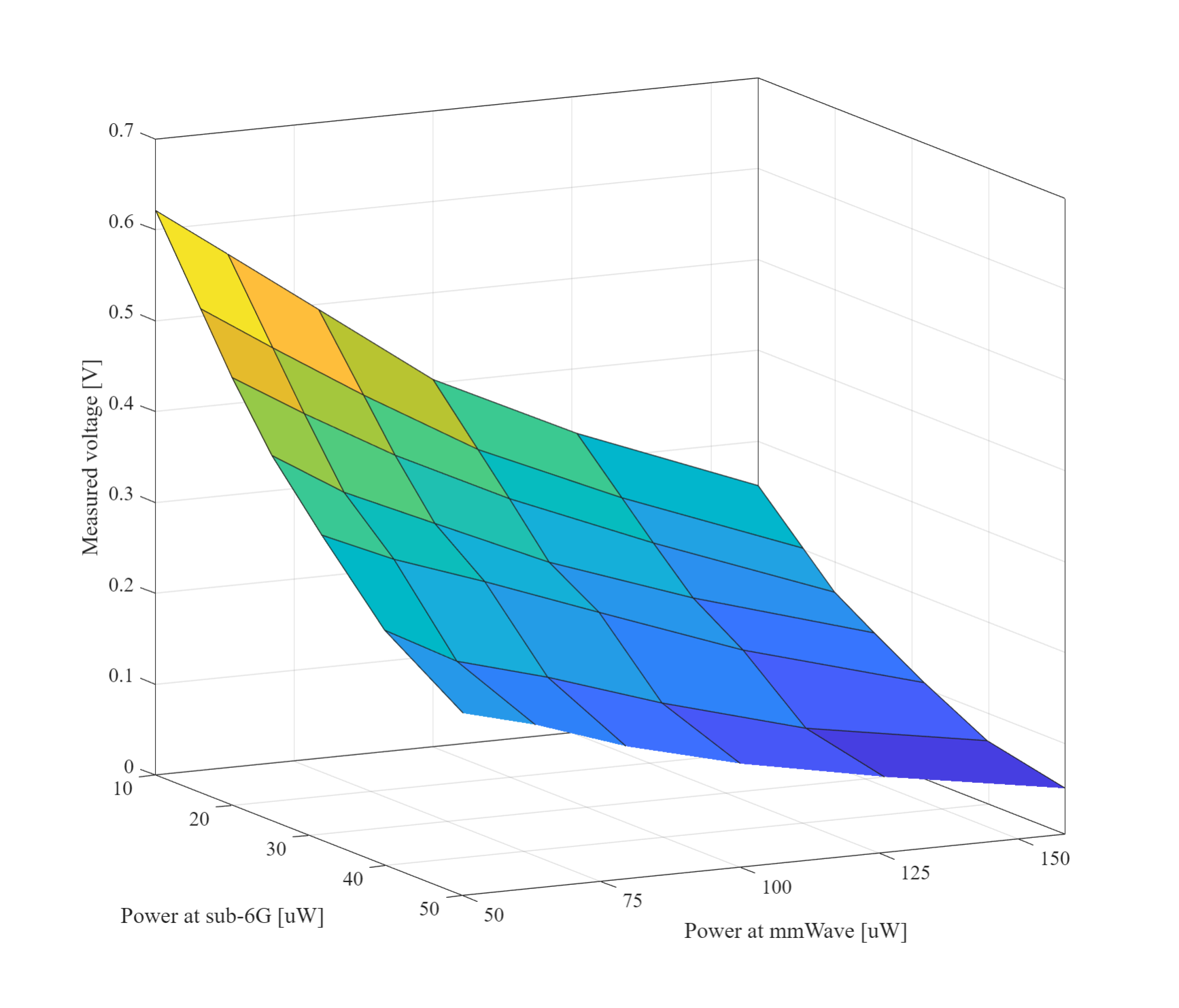}}
    \subfigure[Top view]
    {\includegraphics[width=0.49\columnwidth]{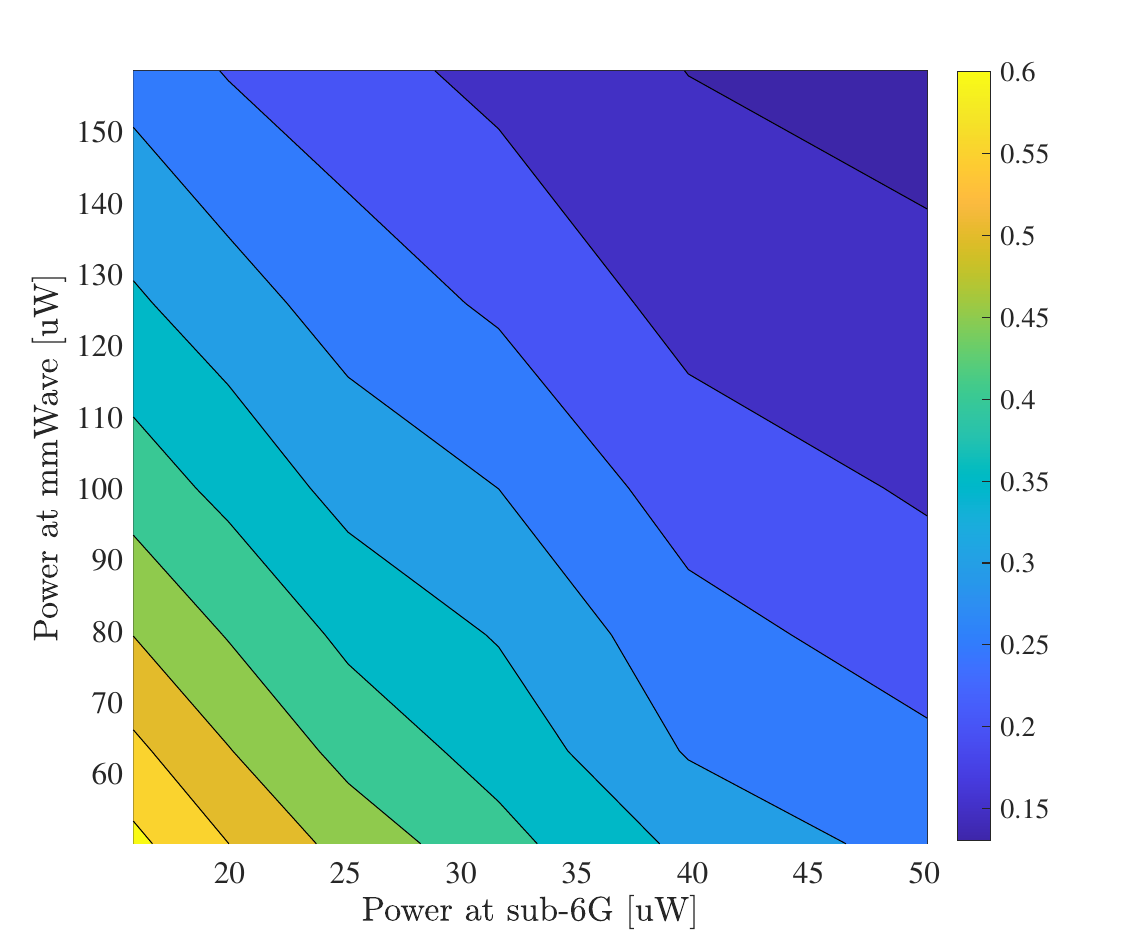}}
    \subfigure[sub-6G]
    {\includegraphics[width=3.3in]{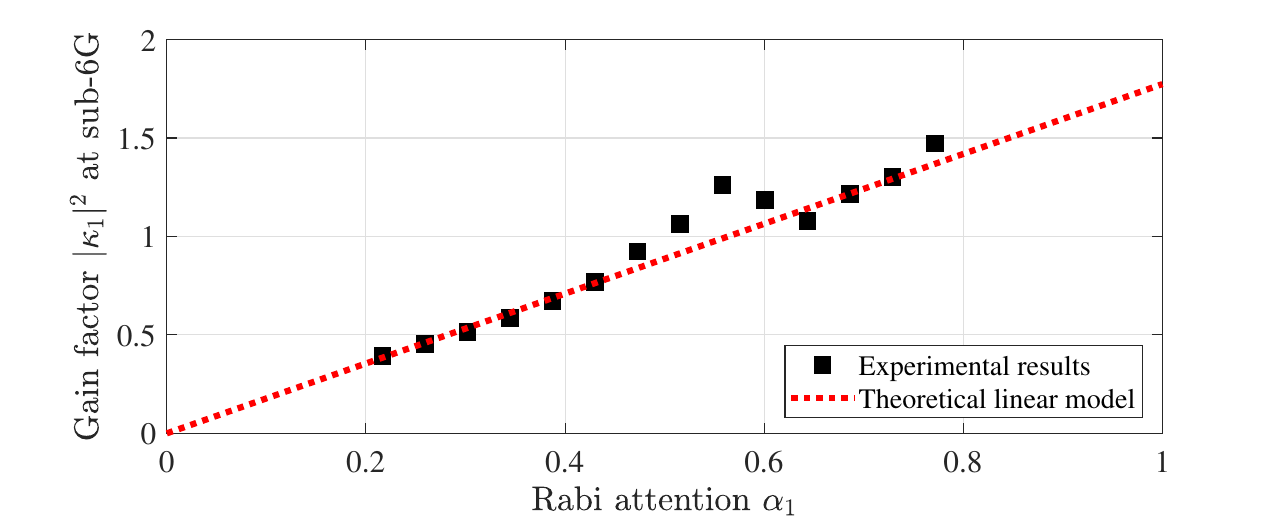}}
    \subfigure[mmWave]
    {\includegraphics[width=3.3in]{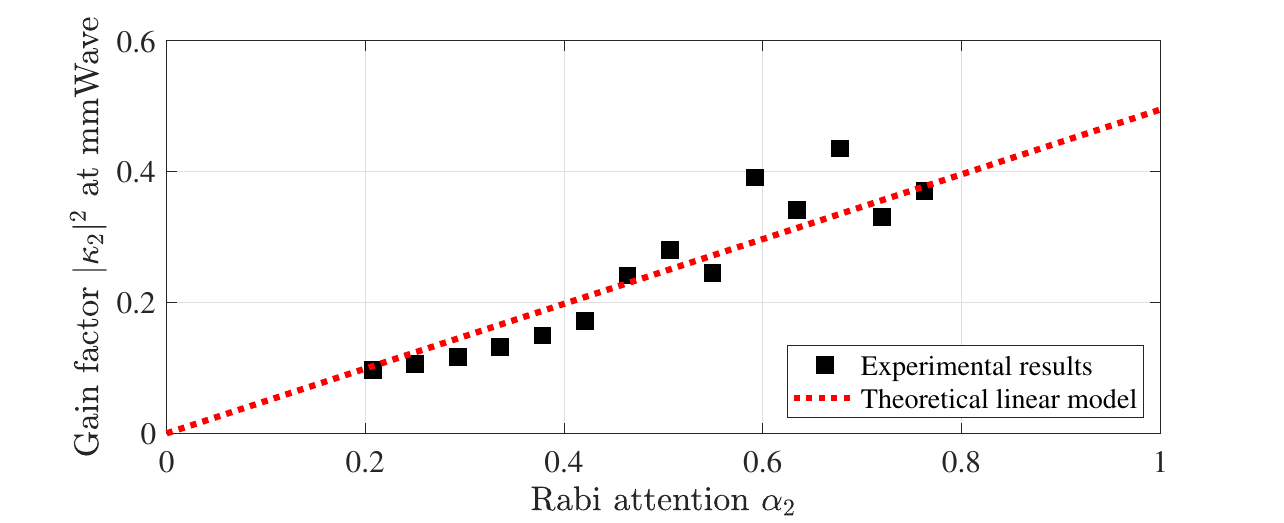}} 
    \caption{\color{black} (a)-(b) Measured output voltage versus transmit powers of the dual-band signals. (c)-(d) Measured band-specific gain factor as a function of Rabi attention. }
	\vspace*{-1em}
	\label{img:exp} 
\end{figure}

Fig.~\ref{img:exp} presents the experiment results.  The measured DC voltage as a function of the transmit powers of the applied dual-band fields is illustrated in subfigures (a) and (b). The voltage decreases with increasing transmit power in either the sub-6G or mmWave band, consistent with the theoretical prediction in Eq.~\eqref{eq:y1}, indicating a monotonic reduction in photocurrent w.r.t Rabi frequency. The top view in Fig.~\ref{img:exp}(b) clearly shows that the contours of the resulting voltage approximate straight lines with a uniform slope. This observation confirms the dependence of DC bias on the weighted sum of transmit powers of reference signals. 
Curve fitting yields the following family of straight-line contours:
\begin{align}
    3.31P_{{\rm r}, 1} + P_{{\rm r}, 2} = W\notag.
\end{align}
{\color{black} The weight 3.31 arises from the combined effect of several physical factors that scale the transmitted power to the Rabi frequency, including the transition dipole moment, the antenna gain, the path loss, and the attenuation at the surface of vapor cell.} 
A slight deviation between the measured and fitted contours is observed in high-power regions (e.g., $P_{{\rm r}, 1}>40\:\mu{\rm W}$ and  $P_{{\rm r}, 2}>130\:\mu{\rm W}$). This is owing to the reduced DC bias under strong reference signals, which increases susceptibility to measurement noise.
Furthermore, subfigures (c) and (d) depict the squared intrinsic gains, $\kappa_1^2$ and $\kappa_2^2$, as functions of the Rabi attentions, determined respectively by $\alpha_1 = \frac{\Omega_{{\rm r}, 1}^2}{\Omega_{{\rm r}, 1}^2 + \Omega_{{\rm r}, 2}^2} = \frac{3.31 P_{{\rm r}, 1}}{3.31 P_{{\rm r}, 1} + P_{{\rm r}, 2}}$  and $\alpha_2 = \frac{\Omega_{{\rm r}, 2}^2}{\Omega_{{\rm r}, 1}^2 + \Omega_{{\rm r}, 2}^2} = \frac{P_{{\rm r}, 2}}{3.31 P_{{\rm r}, 1} + P_{{\rm r}, 2}}$. The intrinsic gains are obtained by evaluating the gradient of DC voltage $V_{\rm r}$ in Fig.~\ref{img:exp}(a) along the contour $3.31 P_{{\rm r}, 1} + P_{{\rm r}, 2} = 130\:{\mu {\rm W}}$, which emulates a fixed effective Rabi frequency. As shown in Fig.~\ref{img:exp}(c)-(d), the experimental data points align closely with the theoretically predicted linear relationships (red curves). {\color{black} The normalized-mean-square fit residuals, evaluated by $\frac{\mathrm{E}({\rm gain}_{\rm measure} - {\rm gain}_{\rm fit})^2}{\mathrm{E}({\rm gain}_{\rm fit})^2}$, for the linear relationships are calculated as 0.0099 and 0.0278 respectively. These small residuals confirm the high quality of the linear fits and demonstrate that, under a fixed effective Rabi frequency, the sensitivity of a RARE to different frequency bands grows linearly with the associated Rabi attentions. }

}

\subsection{Numerical Validation for Signal Model}\label{sec:6B}
In this subsection, we present numerical results to verify the accuracy of the derived transfer function for multi-band RARE systems. 


\subsubsection{Simulation setup}

\begin{table}[]
\caption{Parameters of multi-level electron transitions}
\centering
\begin{tabular}{|c|c|c|c|}
\hline
\rowcolor{lightblue}
Transition                   & Energy level                     & Freq. [GHz] & Dipole moment$^*$ \\ 
\hline
$\ket{3}\rightarrow\ket{4}$  & $60D_{5/2}\rightarrow 61P_{3/2}$ & 3.212                & $2409\:qa_0$                         \\ \hline
$\ket{3}\rightarrow\ket{5}$  & $60D_{5/2}\rightarrow 58F_{7/2}$ & 14.791               & $2013\:qa_0$                         \\ \hline
$\ket{3}\rightarrow\ket{6}$  & $60D_{5/2}\rightarrow 57F_{7/2}$ & 19.883               & $1588\:qa_0$                         \\ \hline
$\ket{3}\rightarrow\ket{7}$  & $60D_{5/2}\rightarrow 62P_{3/2}$ & 30.618               & $736\:qa_0$                         \\ \hline
$\ket{3}\rightarrow\ket{8}$  & $60D_{5/2}\rightarrow 60P_{3/2}$ & 38.858               & $337\:qa_0$                          \\ \hline
$\ket{3}\rightarrow\ket{9}$  & $60D_{5/2}\rightarrow 59F_{7/2}$ & 47.718               & $184\:qa_0$                          \\ \hline
$\ket{3}\rightarrow\ket{10}$ & $60D_{5/2}\rightarrow 56F_{7/2}$ & 56.434               & $260\:qa_0$                          \\ \hline
$\ket{3}\rightarrow\ket{11}$ & $60D_{5/2}\rightarrow 63P_{3/2}$ & 62.756               & $229\:qa_0$                          \\ \hline
\end{tabular}

\begin{tablenotes}
\footnotesize
\item $^*$$a_0 = 52.9{\rm pm}$ denotes the Bohr radius~\cite{SIBALIC2017319}.
\end{tablenotes}
\label{tab2}
\end{table}

Unless otherwise specified, the following simulation settings are adopted. Cesium atoms are confined in a glass vapor cell of length $L = 5\:{\rm cm}$ with an atomic density of $N_0 = 5\times10^{8}\:{\rm cm}^{-3}$.
A probe laser at $\sim$852~nm with Rabi frequency $\Omega_{\rm p} = 2\pi\times 6\:{\rm MHz}$ and a coupling laser at $\sim$510~nm with Rabi frequency $\Omega_{\rm c} = 2\pi\times 10\:{\rm MHz}$ drive the electron transitions $6S_{1/2}\rightarrow 6P_{3/2}\rightarrow 60D_{5/2}$. The decay rate is set to $\gamma_2 = 2\pi\times 5\:{\rm MHz}$. 
{\color{black} The complete parameters for the considered Rydberg electron transitions $\ket{e_0}\rightarrow \ket{e_n}$, including their energy levels and transition dipole moments, are provided in Table~\ref{tab2}. These transitions resonate with $N = 8$ disjoint frequency bands spanning from 3.212 GHz to 62.756 GHz.  } The quantum efficiency is $\eta = 0.8$.

\subsubsection{Numerical results}
\begin{figure}
	\centering
	\includegraphics[width=1\columnwidth]{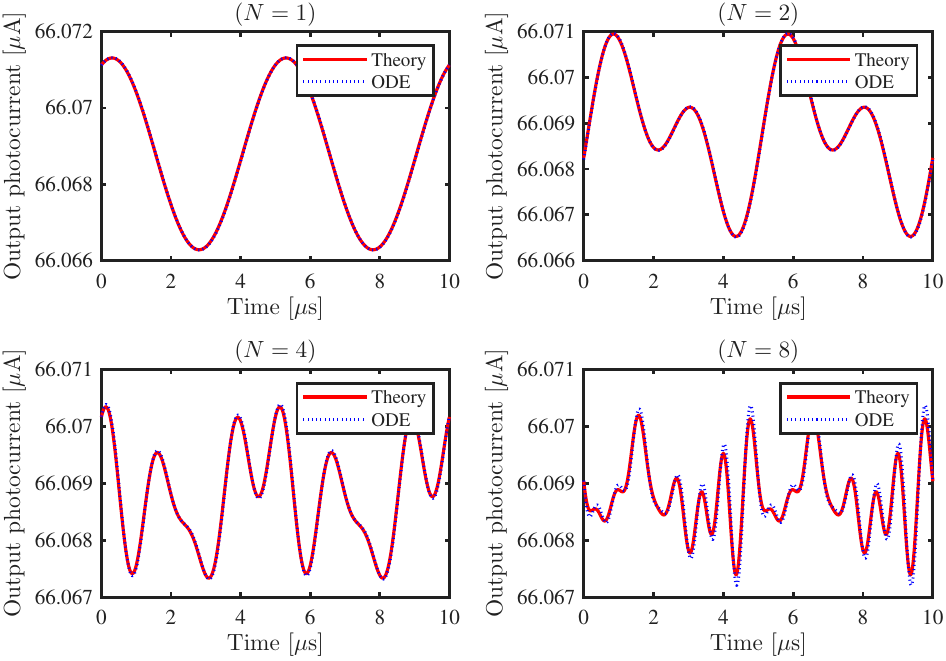}
	\caption{Waveforms of output photocurrent for $N = 1,2,4,8$.} 
	\vspace*{-1em}
	\label{img:waveform}
\end{figure}

Fig.~\ref{img:waveform} illustrates the alignment of waveforms between the theoretical signal model acquired in \eqref{eq:y2}, labeled as ``Theory", and the numerically exact solution, labeled as ``ODE". The exact solutions are obtained by using
an ordinary differential equation (ODE) solver (ODE89 in MATLAB) 
to address the Lindblad master equation given in \eqref{eq:lindblad}.  The number of frequency bands is set to $N = \{1,2,4,8\}$. 
A time window of $1\:{\rm ms}$ is selected to observe the waveforms, during which the amplitude of wireless signals in each band remains constant. The Rabi 
frequencies of reference signals are set to $\Omega_{l,n} = 2\pi \times 
\frac{4}{\sqrt{N}}\:
{\rm MHz}, \forall n$, while those of the data signals are given by $\Omega_{{\rm s},n} = a\Omega_{l,n}e^{j2b\pi}$. Here, $a = 10^{-3}$ ensures the strong reference condition $|\Omega_{\mathrm{r},n}| \gg |\Omega_{\mathrm{s},n}|, \forall n$, and $b \sim \mathcal{U}(0,1)$ models a random phase shift. The IFs for each frequency band are selected 
as $\delta_n=2\pi \times 200n\:{\rm kHz},\forall n\in\{1,2,\cdots, N\}$. 
As presented in Fig.~\ref{img:waveform}, the theoretical waveforms are perfectly consistent with the exact solutions for $N = 1, 2, 4$. A minor deviation is observed only for $N = 8$. These results confirm the accuracy of our theoretical model for both dual-band and multi-band RARE systems.

\begin{figure}[t!]\color{black}
    \centering
    \subfigure[NMSE versus $\delta_{\rm max}$]
    {\includegraphics[width=0.49\columnwidth]{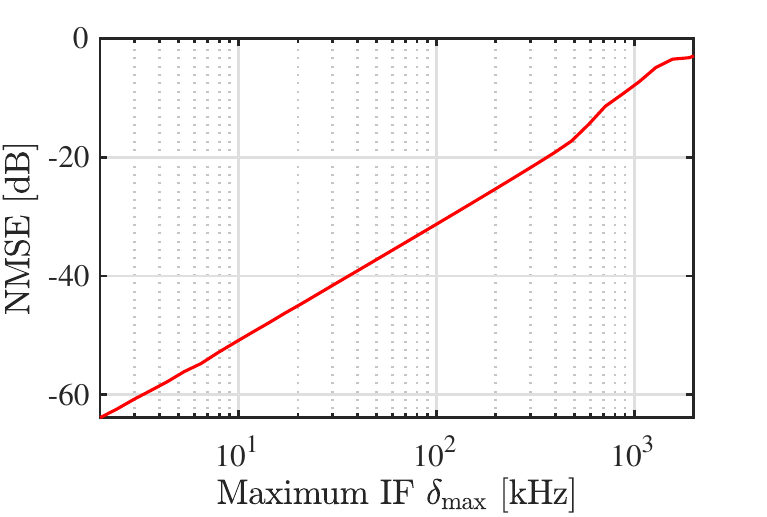}}
    \subfigure[NMSE versus ${|\Omega_{{\rm s},n}|}/{|\Omega_{{\rm r},n}|}$]
    {\includegraphics[width=0.49\columnwidth]{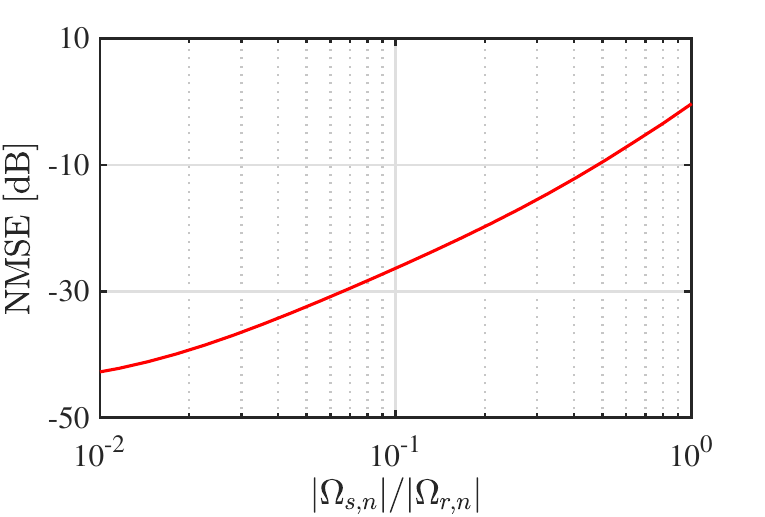}}
    \caption{\color{black} Influence of (a) the maximum IF $\delta_{\rm max}$ and (b) the ratio of Rabi frequency $a = {|\Omega_{{\rm s},n}|}/{|\Omega_{{\rm r},n}|}$ on the NMSE between the theoretical and numerical waveforms. }
	\vspace*{-1em}
	\label{img:NMSE_waveform1} 
\end{figure}
{\color{black} 
To understand the slight discrepancy in the $N = 8$ case, we analyze in  Fig.~\ref{img:NMSE_waveform1} the influence of the maximum IF $\delta_{\max} \overset{\Delta}{=}\delta_N$ and the ratio of Rabi frequency $a = {|\Omega_{{\rm s},n}|}/{|\Omega_{{\rm r},n}|}$ on the NMSE between the theoretical and numerical waveforms. Both subfigures apply $N = 8$ frequency bands. Specifically, the steady-state solution in Theorem 1 requires the time derivative of the density matrix to vanish, $\partial \boldsymbol{\rho} / \partial t \approx 0$. This condition hinges critically on two factors:  a small maximum IF $\delta_{\max}$ and weak data signal strength satisfying $|\Omega_{{\rm s},n}|\ll|\Omega_{{\rm r}, n}|$. First, $\delta_{\rm max}$  must lie within the maximum instantaneous bandwidth (IB) of the RARE system to allow the quantum state $\rho_{12}$ adequate time to reach steady state. According to~\cite{PhysRevA.69.063801}, the IB of a RARE is approximately half the decay rate: ${\rm IB} \approx \frac{\gamma_2}{2} = 2\pi\times 2.5\:{\rm MHz}$. As $\delta_{\max}$ increases from $1\:{\rm kHz}$ to $2\:{\rm MHz}$ in Fig.~\ref{img:NMSE_waveform1}(a), the NMSE rises from -60 dB to 0 dB. 
To maintain an NMSE below -10 dB, the maximum IF cannot exceed 1 MHz. This implies a requirement on the maximum IF: $\delta_{\rm max} \le 1\:{\rm MHz} \approx 0.4\times{\rm IB} = 0.2\gamma_2$. 
This result explains the perfect alignment between the theoretical and numerical waveforms in Fig.~\ref{img:waveform} for $N\le 4$, where $\delta_{\rm max} = \delta_{4} = 2\pi\times 800\:{\rm kHz} <0.2\gamma_2 = 2\pi\times 1\:{\rm MHz}$, as well as the slight deviation for $N = 8$, where $\delta_{\rm max} = \delta_{8} = 2\pi\times 1.6\:{\rm MHz} > 0.2\gamma_2$. The second factor supporting the steady-state solution is that the data signal must be a small perturbation compared to the reference signal. Under this condition, the overall Rabi frequency keeps nearly constant,  $\Omega_{n} \approx \Omega_{{\rm r}, n}$, and the linearization in \eqref{eq:y2} remains valid. In Fig.~\ref{img:NMSE_waveform1} (b), with $\delta_{\max}$ fixed at $2\pi \times 20\:\mathrm{kHz}$, we increase the Rabi frequency ratio $a = |\Omega_{{\rm s}, n}|/|\Omega_{{\rm r}, n}|$ from $10^{-2}$ to $10^0$ to examine the feasible range of data signal strength. The NMSE is observed to rise from -40~dB to over 0~dB, underscoring the necessity of a weak data signal. To ensure a low NMSE, we recommend maintaining ${|\Omega_{\mathrm{s},n}|} / {|\Omega_{\mathrm{r},n}|} < 0.1$ in practical implementations.

\begin{figure}[t!]\color{black}
    \centering
        \subfigure[NMSE versus $\Delta_{1}$]
    {\includegraphics[width=0.49\columnwidth]{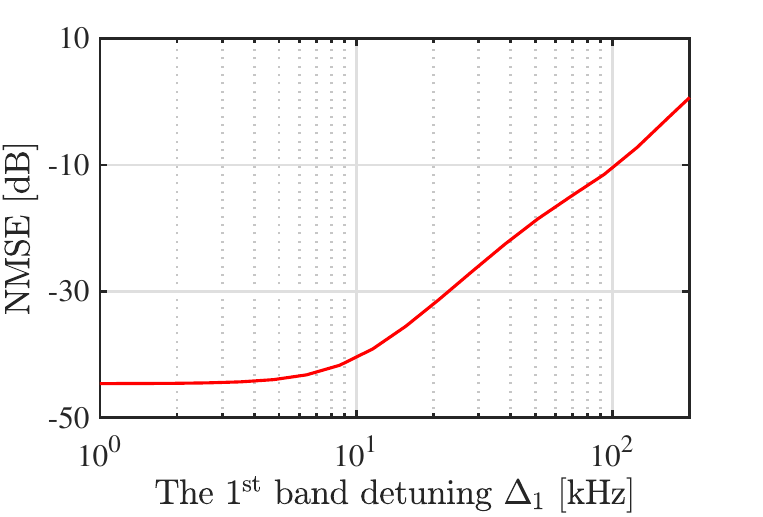}}
    \subfigure[NMSE versus $\Delta_{2}$]
    {\includegraphics[width=0.49\columnwidth]{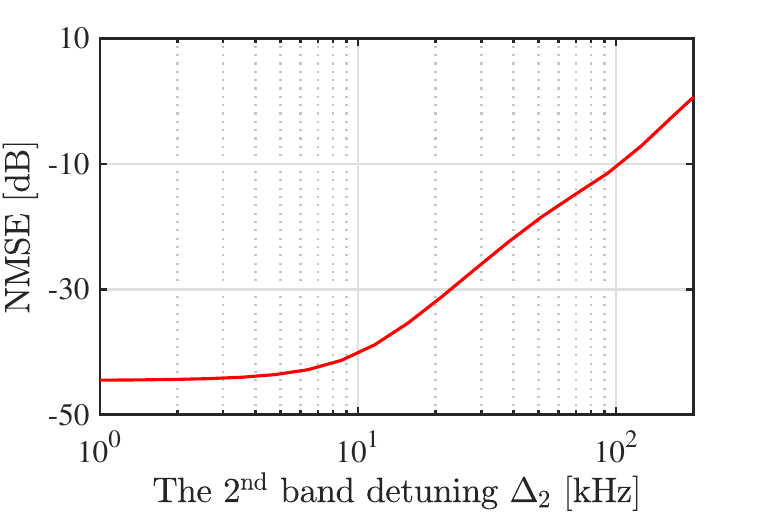}} 
    \subfigure[NMSE versus $\Delta_{\rm p}$]
    {\includegraphics[width=0.49\columnwidth]{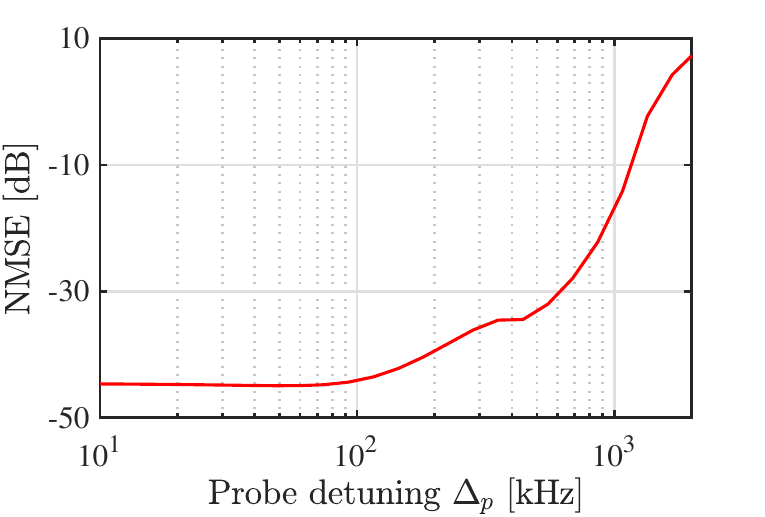}}
    \subfigure[NMSE versus $\Delta_{\rm c}$]
    {\includegraphics[width=0.49\columnwidth]{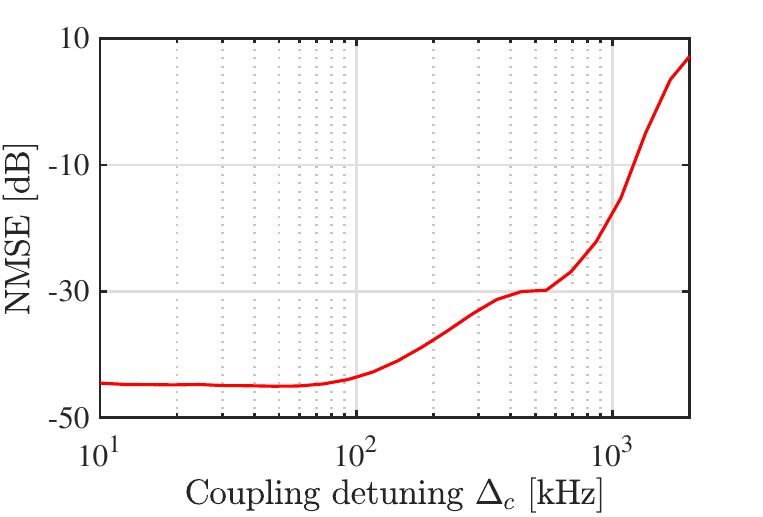}} 
    \caption{\color{black} Influence of frequency detunings of (a) the first band, (b) the second band, (c) the probe laser, and (d) the coupling laser on the NMSE of waveforms.}
	\vspace*{-1em}
	\label{img:NMSE_waveform2} 
\end{figure}

Our theoretical model is established under the assumption that the fields are resonant with their corresponding transition frequencies. 
To evaluate its robustness under non-resonant scenarios, Fig.~\ref{img:NMSE_waveform2} investigates the effect of frequency detunings on the NMSE between theoretical and numerical waveforms. Here, we set $a = 10^{-3}$, $N = 2$, and $\delta_{\rm max} = 2\pi\times 20\:{\rm kHz}$.  Let $\Delta_1 = \omega_1 - \omega_{e_0e_1}$ and $\Delta_2 = \omega_2 - \omega_{e_0e_2}$ denote the detunings of the first and second RF bands, respectively, and $\Delta_{\rm p} = \omega_{\rm p} - \omega_{g_0g_1}$ and $\Delta_{\rm c} = \omega_{\rm c} - \omega_{g_1e_0}$ those of the probe and coupling lasers. Such detunings may arise from frequency drifts in the field sources. As shown in Fig.~\ref{img:NMSE_waveform2}, the NMSE remains below -10~dB provided that the microwave detunings $\Delta_1$ and $\Delta_2$ are within $100\:\mathrm{kHz}$, and the laser detunings $\Delta_{\rm p}$ and $\Delta_{\rm p}$ do not exceed $1\:\mathrm{MHz}$. 
These findings confirm that the proposed model maintains good performance under small frequency offsets typically encountered in experiments, thereby clarifying its practical scope under non-ideal frequency conditions.

}



\subsection{Numerical Validation for Multi-band CommunSense Performance}\label{sec:6C}
We now evaluate multi-band CommunSense performance by comparing the multi-band RARE with classical receivers. The communication performance is quantified by the SE, while the sensing performance is assessed using the NCRLB. 

\subsubsection{Simulation setup}
We consider a dual-band CommunSense system. Two Rydberg states, $61P_{3/2}$ and $58F_{7/2}$, listed in Table~I are selected to detect midband signals at $\omega_1 = 2\pi\times 14.791\:{\rm GHz}$ and mmWave signals at $\omega_2 = 2\pi\times 30.628\:{\rm GHz}$. The Rabi attentions for the two bands are expressed as $\alpha_1 = \alpha$ and $\alpha_2 = 1 - \alpha$, with 
$\alpha$ ranging from 0 to 1. The powers of $\sim$852nm probe laser and $\sim$510nm coupling laser are set to 100$\:{\rm \mu W}$ and $40\:{\rm mW}$, respectively, and the decay rate is $\gamma_2 = 2\pi\times 5$~MHz. 
The IFs and occupied bandwidths of the dual-band signals are set to $\delta_1 = 
2\pi\times 100\:{\rm kHz}$,  $B_1 = 50\:{\rm kHz}$ and $\delta_2 = 2\pi\times 
200\:{\rm kHz}$,  $B_2 = 50\:{\rm kHz}$. Both wireless data sources deliver a transmit power of $0\:{\rm dBm}$. The link distance is randomly sampled between 50~m and 1000~m. {\color{black} For comparison, a dual-band classical receiver system is emulated by deploying two dedicated receivers operating at  $\omega_1 =2\pi\times 14.791\:{\rm GHz}$ (midband) and $\omega_2 = 2\pi\times30.628\:{\rm GHz}$ (mmWave). These are labeled as ``CR1'' and ``CR2'', respectively. 
Each receiver is equipped with a resonant half-wavelength dipole antenna (lengths of 10.1~mm and 4.9~mm) and a dedicated front-end circuit. The system parameters are defined as follows: antenna gain is 2~dBi, the system operates at a noise equivalent temperature of 300~Kelvin, and the signal bandwidth is set to 50 kHz, identical to that of the Rydberg receiver for a direct performance comparison. }

\subsubsection{Numerical results}

\begin{figure}[t!]
    \centering
    \subfigure[Communication performance]
    {\includegraphics[width=0.49\columnwidth]{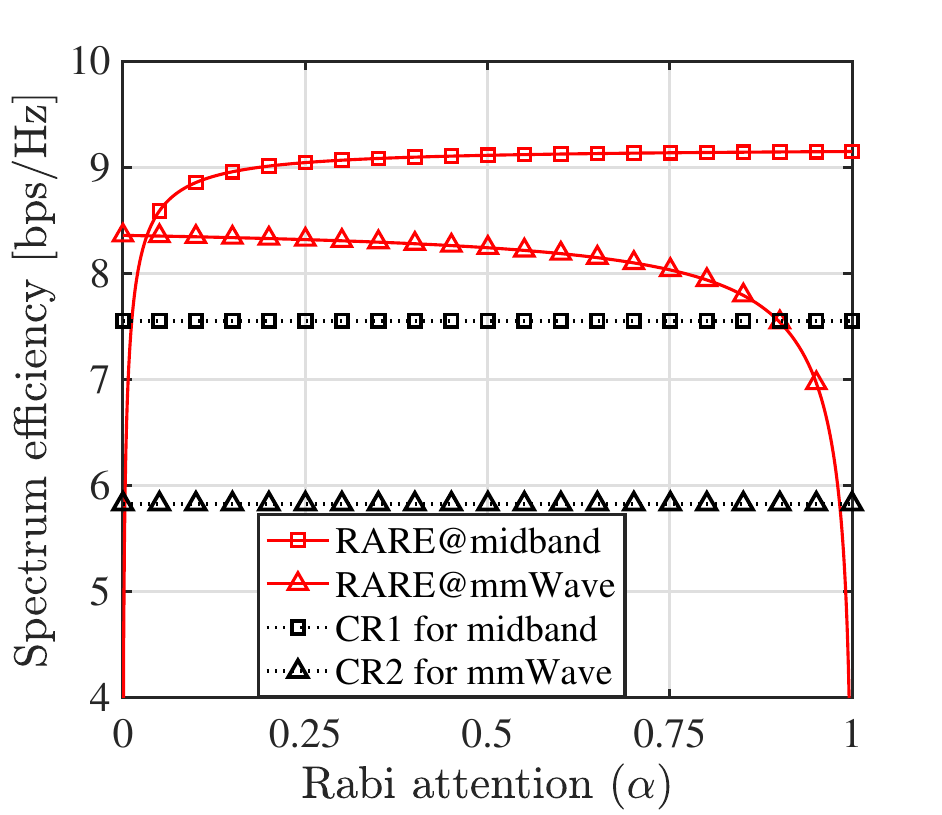}}
    \subfigure[Sensing performance]
    {\includegraphics[width=0.49\columnwidth]{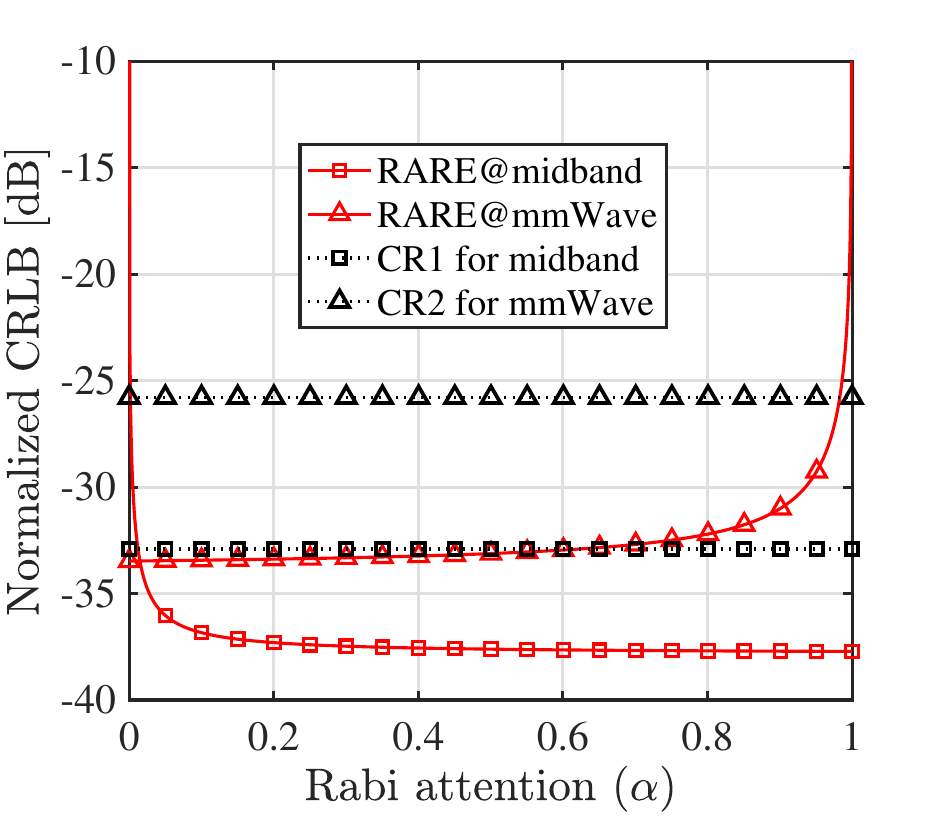}}
    \caption{CommunSense performance at different frequency bands for different settings of Rabi attention, $\alpha$.}
	\vspace*{-1em}
	\label{img:attention} 
\end{figure}

Fig.~\ref{img:attention} evaluates the CommunSense performance under 
varying Rabi attentions.  As shown, the midband
performance achieved by RARE continuously 
improves with increasing $\alpha$, whereas the mmWave performance
declines correspondingly.  This trend stems from the attention mechanism,  where $\alpha$ governs the fraction of sensitivity allocated to midband.
Notably, when 
$\alpha = 1$,  the system reduces to a conventional single-band RARE.
In this case, the CommunSense performance at midband is maximized, achieving an SE of $9.2\:{\rm bps/Hz}$ and an NCRLB of 
$-38\:{\rm 
dB}$, while mmWave-band CommunSense is completely disabled. Conversely, 
setting $\alpha = 0.5$ yields a balanced configuration: the midband performance experiences only marginal degradation (a $0.1\:{\rm bps/Hz}$ reduction in SE and 
$0.2\:{\rm dB}$ increase in NCRLB), whereas the mmWave band shows substantial improvement: SE rises from $0\:{\rm bps/Hz}$ to  
$8.2\:{\rm bps/Hz}$ and the 
NCRLB drops from infinity to $-33\:{\rm dB}$. This example highlights the superiority of the multi-band RARE over its single-band counterpart in handling heterogeneous CommunSense tasks. 
Moreover, when the Rabi attention is properly configured, the multi-band RARE achieves significantly better performance than the two classical receivers. For instance,  at $\alpha = 
0.4$, the multi-band RARE yields an average improvement of $2\:{\rm bps/Hz}$ in SE and a 
$6\:{\rm dB}$ reduction in 
NCRLB. 
 It is worth emphasizing that all these multi-band CommunSense gains are accomplished within a single integrated RARE device.

\begin{figure}[t!]
\centering
\subfigure[Communication performance]
{\includegraphics[width=0.49\columnwidth]{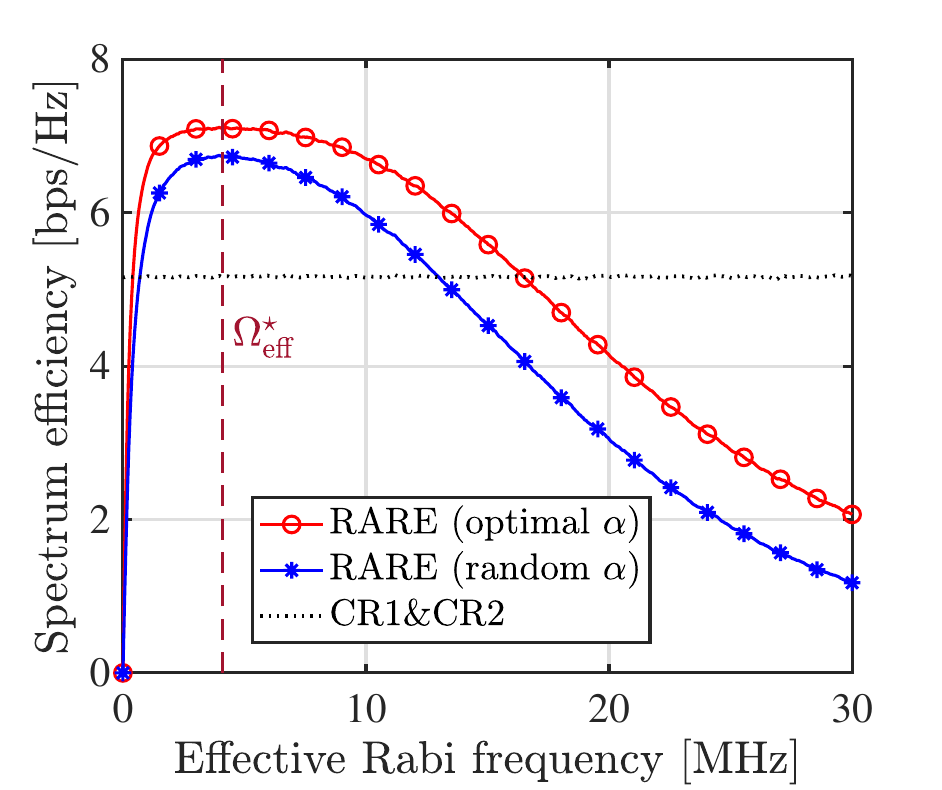}}
\subfigure[Sensing performance]
{\includegraphics[width=0.49\columnwidth]{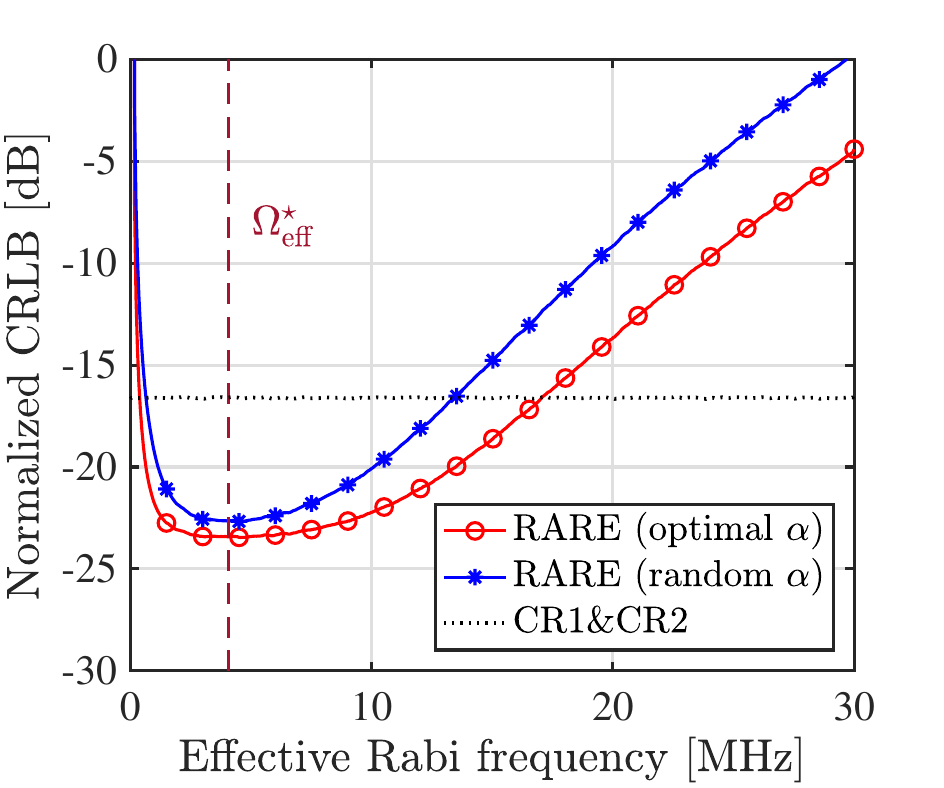}}
\caption{Performance comparison between the multi-band RAREs and classic receivers for 
	different settings of effective Rabi frequency, $\Omega_{\rm eff}$. }
	\vspace*{-1em}
	\label{img:sumsquare} 
    \end{figure}


Then, we evaluate in Fig.~\ref{img:sumsquare} the influence of the effective Rabi frequency $\Omega_{\rm eff}$ on (a) SE and (b) NCRLB, both averaged over the two frequency bands. The following configurations are compared: {\color{black}(1)  a multi-band classical receiver implemented by the integration of two single-band classical receivers, denoted as ``CR1$\&$CR2''}, (2) a multi-band RARE with Rabi attention weights randomly sampled as $\alpha \sim \mathcal{U}(0,0.5)$, and (3) a multi-band RARE with the optimally configured Rabi attention. As $\Omega_{\rm eff}$ increases, the performance of the RARE systems exhibits a characteristic rise followed by a decline. This behavior can be explained through the underlying electron transition dynamics.
When the effective Rabi frequency is either too low or too high, the electron population becomes predominantly localized in either the initial or final Rydberg states. In both cases, the system responds poorly to the weak wireless data signal.  Optimal sensitivity is achieved only at intermediate Rabi frequencies, where a moderate number of electrons occupy the final Rydberg states, enabling effective signal detection.
Notably, the optimal effective Rabi frequency, $\Omega_{\rm eff}^\star$, can consistently capture the most sensitive point, regardless of the specific Rabi attention weights or the CommunSense services considered. This universality arises because $\Omega_{\rm eff}^\star$ governs the global gain $\varrho_0$, which amplifies signals uniformly across all frequency bands, thereby enhancing overall system sensitivity.
Using the optimal effective Rabi frequency and Rabi attentions,  the RARE outperforms the classical receiver benchmark by $1.9\:{\rm bps/Hz}$ in SE and $7\:{\rm dB}$ in NCRLB. 
These results underscore the critical importance of optimizing $\Omega_{\rm eff}$  to maximize the sensitivity of a RARE.


\begin{figure}[t!]
\centering
\subfigure[Communication performance]
{\includegraphics[width=0.49\columnwidth]{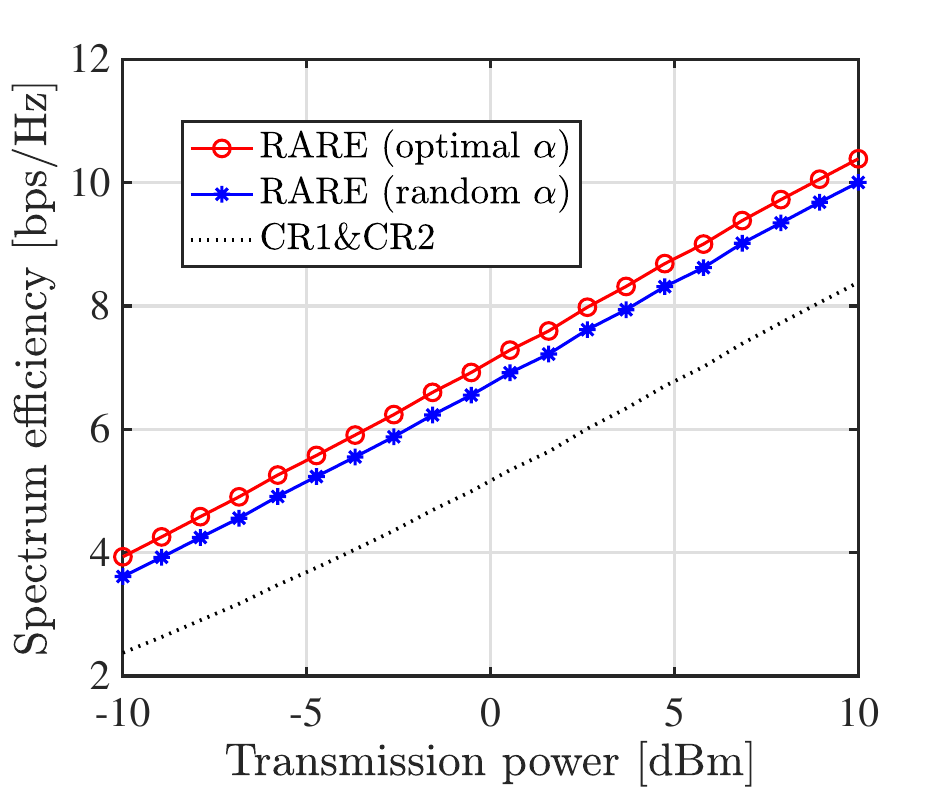}}
\subfigure[Sensing performance]
{\includegraphics[width=0.49\columnwidth]{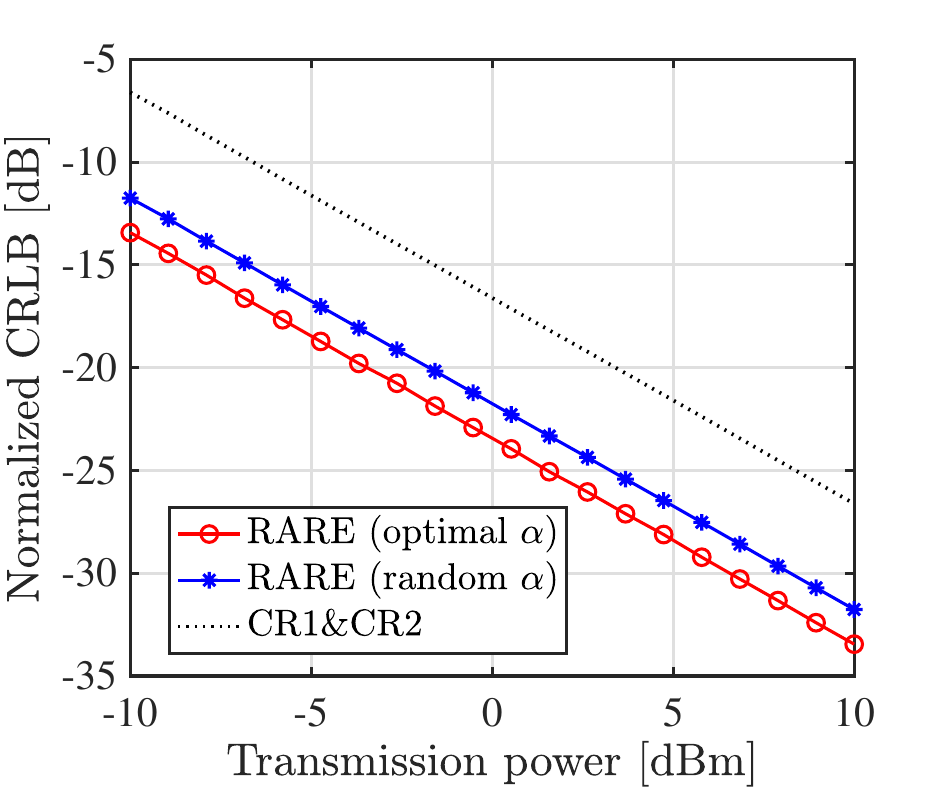}}
\caption{Performance comparison between the multi-band RAREs and classic receivers for 
		different settings of the transmission power, $P_{{\rm s}, n}$. }
	\vspace*{-1em}
	\label{img:ptx} 
    \end{figure}
Fig.~\ref{img:ptx} depicts the CommunSense 
performance as a function of 
the transmission power of data signals, $P_{{\rm t} , n}$, 
varied from $-10\:{\rm dBm}$ to 
$10\:{\rm dBm}$. 
The optimal effective Rabi frequency is adopted for both ``RARE 
(optimal $\alpha$)" and ``RARE (random $\alpha$)". As observed, the multi-band RARE systems significantly outperform the classical counterpart across the entire considered power range. Furthermore, compared to the ``RARE (random $\alpha$)" baseline, the optimized Rabi attention weights further 
improve the SE by $0.5\:{\rm bps/Hz}$ and reduce the NCRLB by $2\:{\rm dB}$. 
These results highlight the critical role of optimizing Rabi attention weights in tailoring RARE systems for specific communication and sensing tasks.

\section{Conclusions and Future Directions}\label{sec:7}
In this paper, we have mathematically characterized the physical mechanisms of multi-band RAREs. {\color{black} 
An analytical transfer function has been derived with explicit sensitivity allocation.} Our analysis reveals that a multi-band RARE functions as both a multi-band atomic mixer and a multi-band atomic amplifier. 
Additionally, the attention mechanism of the multi-band atomic amplifier has been discovered, which was demonstrated by real-world experiments. We have also presented optimal design for the atomic amplifier to maximize the receiver sensitivity and the overall CommunSense performance.

Despite these contributions, the current framework proposed in this paper represents a starting point towards quantum sensor enabled CommunSense. Some unresolved issues remain further investigations. For example, our framework distinguishes frequency bands by allocating them to different intermediate frequencies. An alternative approach is to integrate multi-band RARE with atomic MIMO technology~\cite{AtomicPrecoding_Cui2024, AtomicMIMO_Cui2025}, which enables spatial orthogonality between frequency bands.   In such configurations, variations in reference signal propagation paths can lead to non-uniform Rabi frequencies across receiving antennas. This effect results in spatially heterogeneous sensitivity patterns across frequency bands, which necessitates systematic investigation to enable efficient atomic multi-band MIMO techniques. 
    Moreover, this work assumes the reference signals are co-polarized with the probe and coupling lasers. This can be achieved by adjusting the half-wave plates (HWPs) in each optical path to produce a symmetrical EIT spectrum with only two peaks~\cite{RydbergPolarization_Sedlacek2013, Jiao_2019, 10.1063/5.0146768}. For misaligned polarizations, multiple $m_j$-dependent transition pathways can be activated. This would result in three EIT peaks for a four-level quantum system, and potentially multiple peaks for an $(N+3)$-level quantum system. A full analysis of the quantum response under arbitrary polarization conditions would require incorporating fine‑structure terms into the multi‑band Hamiltonian and is left for future study. 




\appendix{
{\color{black}
\subsection{Steady-State Solution of $\boldsymbol{\rho}$ to Equation \eqref{eq:lindblad}} \label{app:solution}
The steady-state solution for the density matrix $\boldsymbol{\rho}$ is derived from \eqref{eq:lindblad}. The elements corresponding to the first three rows and columns are:
\begin{align}
\rho_{11} & = \frac{(\Omega_{\rm p}^2+\gamma_2^2)\sum_{m=1}^N\Omega_m^2 + \Omega_{\rm c}^2\Omega_{\rm p}^2}{(\gamma_2^2+2\Omega_{\rm p}^2)\sum_{m=1}^N\Omega_m^2 + 2\Omega_{\rm p}^2(\Omega_{\rm c}^2 + \Omega_{\rm p}^2)}, \notag\\
\rho_{12} &= j\frac{\gamma_2\Omega_{\rm p}\sum_{m=1}^N\Omega_m^2}{(\gamma_2^2+2\Omega_{\rm p}^2)\sum_{m=1}^N\Omega_m^2 + 2\Omega_{\rm p}^2(\Omega_{\rm c}^2 + \Omega_{\rm p}^2)}, \notag\\
\rho_{13} &= - \frac{\Omega_{\rm c}\Omega_{\rm p}^3}{(\gamma_2^2+2\Omega_{\rm p}^2)\sum_{m=1}^N\Omega_m^2 + 2\Omega_{\rm p}^2(\Omega_{\rm c}^2 + \Omega_{\rm p}^2)},\notag\\
\rho_{22} & = \frac{\Omega_{\rm p}^2\sum_{m=1}^N\Omega_m^2} {(\gamma_2^2+2\Omega_{\rm p}^2)\sum_{m=1}^N\Omega_m^2 + 2\Omega_{\rm p}^2(\Omega_{\rm c}^2 + \Omega_{\rm p}^2)},\notag\\
\rho_{23} & = 0,\notag\\
\rho_{33} & = \frac{\Omega_{\rm p}^4} {(\gamma_2^2+2\Omega_{\rm p}^2)\sum_{m=1}^N\Omega_m^2 + 2\Omega_{\rm p}^2(\Omega_{\rm c}^2 + \Omega_{\rm p}^2)}. \notag
\end{align}
The remaining elements, for indices $n, n' \in \{1,\dots, N\}$, are given by:
\begin{align}
&\rho_{1,n + 3} = -j \frac{\gamma_2\Omega_{n}\Omega_{\rm c}\Omega_{\rm p}}{(\gamma_2^2+2\Omega_{\rm p}^2)\sum_{m=1}^N\Omega_m^2 + 2\Omega_{\rm p}^2(\Omega_{\rm c}^2 + \Omega_{\rm p}^2)},\notag\\
&\rho_{2, n + 3}  =  - \frac{\Omega_{n}\Omega_{\rm c}\Omega_{\rm p}^2}{(\gamma_2^2+2\Omega_{\rm p}^2)\sum_{m=1}^N\Omega_m^2 + 2\Omega_{\rm p}^2(\Omega_{\rm c}^2 + \Omega_{\rm p}^2)},\notag\\
&\rho_{3,n+3}=0, \notag\\
&\rho_{n+3, n'+3}   \notag \\ &=\frac{\Omega_{n}\Omega_{n'}\Omega_{\rm p}^2(\Omega_{\rm p}^2+\Omega_{\rm c}^2)} {(\gamma_2^2+2\Omega_{\rm p}^2)(\sum_{m=1}^N\Omega_m^2)^2 + 2\Omega_{\rm p}^2(\Omega_{\rm c}^2 + \Omega_{\rm p}^2)\sum_{m=1}^N\Omega_m^2}.\notag
\end{align}
Furthermore, since $\boldsymbol{\rho}$ is Hermitian, elements below the main diagonal are obtained by complex conjugation: $\rho_{n,n'} = \rho_{n',n}^*$ for $n > n'$. 
}

{\color{black}
\subsection{Derivation of The Power of Extrinsic Noise}\label{app:bbr}
The number of modes per unit volume per unit bandwidth is given by $\frac{8\pi f^2_n}{c^3} = \frac{2\omega_n^2}{\pi c^3}$. The average field energy per mode due to the blackbody radiation and vacuum fluctuation is $\left(\frac{1}{2} + n_{\rm th}\right)\hbar \omega_n$, where the term $\frac{1}{2}\hbar\omega_n$ arises from the vacuum fluctuation and $n_{\rm th} = {1}/{\left(e^{\hbar \omega_n/k_B T_{\rm a}} - 1\right)}$ is the Bose-Einstein distribution. The product of the number of modes and the energy per mode yields the spectral energy density (energy per unit volume per unit frequency):
\begin{align}
    u_{\rm ext}(\omega_n) = \frac{2\omega_n^2}{\pi c^3} \cdot \left(\frac{1}{2} + n_{\rm th}\right)\hbar \omega_n = \frac{\hbar\omega_n^3}{\pi c^3}(2n_{\rm th} + 1).
\end{align}
For a propagating electromagnetic wave, the total time-averaged energy density is equally divided between its electric and magnetic field components. The energy density due to the electric field alone is $u_{E}(\omega_n) = \frac{1}{2}\epsilon_0 \langle{E^2_{\rm ext}}\rangle$. Therefore, the spectral density of the electric field energy is half of the total spectral energy density: 
$\frac{{\rm d}u_{E}}{{\rm d} f} = \frac{1}{2}\epsilon_0 \frac{{\rm d}\langle{E^2_{\rm ext}}\rangle}{{\rm d} f} = \frac{1}{2}u_{\rm ext}(\omega_n)$. This yields the field spectral density: 
\begin{align}\color{black}
    \frac{{\rm d}\langle{E^2_{\rm ext}}\rangle}{{\rm d} f} = \frac{\hbar\omega_n^3}{\pi \epsilon_0 c^3}(2n_{\rm th} + 1)\:\:[{\rm V}^2 \cdot {\rm m}^{-2} \cdot {\rm Hz}^{-1}].
\end{align}
{\color{black}
The extrinsic field $E_{\rm ext}(t)$ is added onto the RF field $E_{{\rm RF}}(t)$ in free space, and is also weak compared to the strong reference signal. It is thereby transduced to photocurrent via the same microwave-to-optical conversion process presented in \eqref{eq:mapping}, which is amplified by the intrinsic gain, $\kappa_n$, of RARE as well, resulting in a photocurrent spectral density of $\kappa_n^2\frac{{\rm d}\langle{E^2_{\rm ext}}\rangle}{{\rm d} f}$. Consequently, the power of extrinsic noise at band $n$ over a bandwidth of $B_n$ is derived as $\sigma^2_{{\rm E},n} = \frac{\kappa_n^2B_n\hbar\omega_n^3}{\pi \epsilon_0 c^3}(2n_{\rm th} + 1)$.
}

\subsection{Proof of Theorem 2}\label{app:proof_sumsquare}
Problem \eqref{eq:P0} is equivalent to maximizing the function $f(\Omega_{\rm eff}) = \exp\left(-\frac{\chi_0\Omega_{\rm eff}^2}{\Omega_{\rm eff}^2 + \Gamma^2}\right)\frac{\Omega_{\rm eff}^2}{(\Omega_{\rm eff}^2 + \Gamma^2)^4}$ with $\Omega_{\rm eff}^2 > 0$. To solve it, we introduce the variable transformation $x = \frac{\Omega_{\rm eff}^2}{\Omega_{\rm eff}^2 + \Gamma^2}$. Then, we get $0\le x < 1$ and $\Omega_{\rm eff}^2 = \frac{x}{1-x}\Gamma^2$. The function $f(\Omega_{\rm eff})$ can thus be rewritten as 
\begin{align}
    f(\Omega_{\rm eff}) = \frac{1}{\Gamma^6} e^{-\chi_0 x}x(1 - x)^3\overset{\Delta}{=} g(x). 
\end{align}
The gradient of $g(x)$ is given as 
\begin{align}
    \frac{{\rm d} g(x)}{{\rm d} x} = \frac{1}{\Gamma^6}e^{-\chi_0 x}(1 - x^2)(\chi_o x^2 - (\chi_0 + 4)x + 1).
\end{align}
By setting $\frac{{\rm d} g(x)}{{\rm d} x} = 0$ and doing some tedious calculations, we can obtain the optimal $x$ that maximizes $g(x)$ in the range $x\in[0,1)$ as 
\begin{align}
x^\star = \frac{\chi_0 + 4 - \sqrt{\chi_0^2 + 4\chi_0 + 16}}{2\chi_0}.
\end{align}
When $0\le x < x^\star $, the gradient $\frac{{\rm d} g(x)}{{\rm d} x}$ is greater than 0 and thus $g(x)$ monotonically increases, while when $x^\star< x < 1 $, we have $\frac{{\rm d} g(x)}{{\rm d} x} < 0$ and thus $g(x)$ monotonically decreases,

By further invoking the relationship $\Omega_{\rm eff}^2 = \frac{x}{1 - x}\Gamma^2$. The optimal effective Rabi frequency is thus given as
\begin{align}
    \Omega_{\rm eff}^\star = \sqrt{\frac{\chi_0+4 - \sqrt{\chi_0^2 + 4\chi_0 + 16}}{\chi_0 -4 + \sqrt{\chi_0^2 + 4\chi_0 + 16}}}\Gamma.
\end{align}
This completes the proof.

\subsection{Proof of Theorem 3}\label{app:proof_rabiattention}
The Lagrange function of problem \eqref{eq:ASE} is given as 
\begin{align}
    \mathcal{L} = \sum_{n=1}^N\gamma_n\log_2\left(1 + \frac{\beta_n\alpha_n}{\alpha_n + \epsilon_n}\right) + \frac{1}{\nu}\left(1 - \sum_{n = 1}^N\alpha_n\right).
\end{align}
Here, $\nu$ is the Lagrange multiplier and we adopt its reciprocal form $1/\nu$ for ease of later expression. The KKT conditions are given as 
\begin{align}
        \frac{\partial \mathcal{L}}{\partial \mathcal{\alpha}_n} &= \frac{\gamma_n\beta_n\epsilon_n}{((1+\beta_n)\alpha_n + \epsilon_n)(\alpha_n + \epsilon_n)} - \frac{1}{\nu} = 0, \forall n,\label{eq:C1}\\
        \sum_{n = 1}^N\alpha_n &= 1. 
\end{align}
The condition in \eqref{eq:C1} follows a quadratic form. As $\alpha_n \ge 0$, we have $\frac{\gamma_n\beta_n\epsilon_n}{((1+\beta_n)\alpha_n + \epsilon_n)(\alpha_n + \epsilon_n)} > 0$. Therefore, the Lagrange multiplier $\nu$ must be larger than 0 as well. Otherwise, there will be no feasible solutions to \eqref{eq:C1} for all $n$. Given $\nu \ge 0$, the solutions of these quadratic equations should obey 
\begin{align}
    \alpha_n = \frac{-(2+\beta_n)\epsilon_n \pm \sqrt{\beta_n^2\epsilon_n^2 + 4\nu\gamma_n\beta_n(1+\beta_n)\epsilon_n}}{2(1+\beta_n)}. 
\end{align}
By further considering the requirement $\alpha_n \ge 0$, the solution of $\frac{-(2+\beta_n)\epsilon_n - \sqrt{\beta_n^2\epsilon_n^2 + 4\nu\gamma_n\beta_n(1+\beta_n)\epsilon_n}}{2(1+\beta_n)}$ should be discarded as it is always smaller than 0 for all $n$. Besides, the other one needs to be modified as 
\begin{align}
        \alpha_n^\star = \left(\frac{-(2+\beta_n)\epsilon_n + \sqrt{\beta_n^2\epsilon_n^2 + 4\nu^\star\gamma_n\beta_n(1+\beta_n)\epsilon_n}}{2(1+\beta_n)} \right)^+, \notag
\end{align}
where $(x)^+ \overset{\Delta}{=}\max(0, x)$. The optimal Lagrange multiplier $\nu^\star$ is properly selected to satisfy the constraint $\sum_{n=1}^N\alpha_n^\star = 1$. This completes the proof.

}

\bibliographystyle{IEEEtran}
\bibliography{Reference.bib}

\end{document}